\documentclass{IEEEtran}
\usepackage{amssymb}
\usepackage{amsfonts}
\usepackage{amsmath}
\usepackage{algorithmic}
\usepackage{graphicx,cite,multicol,multirow,diagbox,booktabs,array}
\usepackage{subfig}
\usepackage{verbatim}
\usepackage{cuted}
\usepackage{float}
\usepackage{bm}
\usepackage{amsmath}
\usepackage{cases}
\usepackage[ruled,vlined]{algorithm2e}
\usepackage{enumerate}
\usepackage{color}
\usepackage{epstopdf}

\usepackage{etoolbox}
\pretocmd{\maketitle}{%
  \markboth{The definitive version was published in IEEE Transactions on Vehicular Technology, vol. 74, no. 4, pp. 5504-5516, April 2025, doi: 10.1109/TVT.2024.3509613.}{}%
}{}{}

\newtheorem{theorem}{Theorem}
\newtheorem{proof}{\it Proof}

\newtheorem{remark}{Remark}

\ifCLASSOPTIONcompsoc
\fi
\ifCLASSINFOpdf
\else
\fi
\begin{document}
\title{An ADMM-Based Geometric Configuration Optimization in  RSSD-Based Source Localization By UAVs with Spread Angle Constraint}

\author{Xin Cheng, Guangjie Han, \IEEEmembership{Fellow,~IEEE}, Jinlin Peng, Jinfang Jiang, Yu He, Weiqiang Zhu, \\
   Feng Shu, and Jiangzhou Wang, \IEEEmembership{Fellow,~IEEE}

\thanks{
Xin Cheng, Guangjie Han, Jinfang Jiang, and Yu He are with the College of Information Science and Engineering, Hohai University, Changzhou 213200, China. (e-mail: xincstar23@163.com).}
\thanks{
Jinlin Peng is with the Artificial Intelligence Research Center, National Innovation Institute of Defense Technology, Beijing 100071, China.}

\thanks{
Weiqiang Zhu is with the 8511 Research Institute, China Aerospace
Science and Industry Corporation, Nanjing 210007, China.}
\thanks{
Feng Shu is with the School of Information and Communication Engineering,
Hainan University, Haikou 570228, China.}
\thanks{
Jiangzhou Wang is with the School of Engineering, University of Kent, Canterbury CT2 7NT, U.K.}}

\maketitle

\begin{abstract}
  Deploying multiple unmanned aerial vehicles (UAVs) to locate a signal-emitting source covers a wide range of military and civilian applications like rescue and target tracking. It is well known that the UAVs-source (sensors-target) geometry, namely geometric configuration, significantly affects the final localization accuracy. This paper focuses on the geometric configuration optimization for received signal strength difference (RSSD)-based passive source localization by drone swarm. Different from prior works, this paper considers a general measuring condition where the spread angle of drone swarm centered on the source is constrained.
  Subject to this constraint, a geometric configuration optimization problem with the aim of maximizing the determinant of Fisher information matrix (FIM) is formulated. After transforming this problem using matrix theory, an alternating direction method of multipliers (ADMM)-based optimization framework is proposed. To solve the subproblems in this framework, two global optimal solutions based on the Von Neumann matrix trace inequality theorem and majorize-minimize (MM) algorithm are proposed respectively. Finally, the effectiveness as well as the practicality of the proposed ADMM-based optimization algorithm are demonstrated by extensive simulations.
\end{abstract}
\begin{IEEEkeywords}
Source localization, unmanned aerial vehicles, received signal strength difference, geometric configuration optimization, Fisher information matrix, alternating direction method of multipliers, region constraint.
\end{IEEEkeywords}

\IEEEpeerreviewmaketitle

\section{Introduction}
Passive source localization has been applied successfully in many  applications including disaster rescue, mobile communication and electronic warfare \cite{ho2008passive,dogancay2012uav,win2018theoretical}. Recently, with the fast development of unmanned aerial vehicle (UAV) technology, source localization using UAVs has been popular and widely researched \cite{kim2013uav,jiang2020localization,wang2018performance,cheng2021communication}. Compared to traditional sensor carriers like ground vehicles, UAV has the features of small size, fast speed and rapid deployment, thus plays an important role in  arduous localization tasks.

Considering the diversity of sensors, there are several types of localization systems, e.g., direction of arrival (DoA) \cite{shu2018low}, time of arrival (ToA) \cite{xu2011source}, time difference of arrival (TDoA) \cite{ho1993solution}, received signal strength (RSS) \cite{weiss2003accuracy} and received signal strength difference (RSSD) \cite{lohrasbipeydeh2015unknown}. Among these types, the RSS/RSSD-based approach is less sensitive to multi-path interference, thus more suitable for complex localization environment\cite{1458287,zekavat2019handbook}.  Moreover, the RSS/RSSD-based approach is cost-effective and easy-operating since it does not require tight synchronization or calibration like others \cite{patwari2005locating}. Although both measure the signal strength, unlike the RSS-based approach, the RSSD-based approach does not require knowing the transmit power of signal-emitting source \cite{lohrasbipeydeh2014minimax}. Therefore, in some scenarios where the source belongs to a third party or enemy, only the RSSD-based approach is applicable among the two approaches. Due to these advantages, the RSSD-based source localization system by drone swarm is considered in this paper.

A RSSD-based source localization consists of  two stages, i.e., measurement and estimation. Extensive researches focus on designing high-precision estimation method. However, as the prerequisite of estimation, measurement also plays an important role in the final localization precision. How to make measurement becomes a worth-exploring problem. It is well known that for any unbiased estimate, the lower bound of its error can be represented by the Cram\'{e}r-Rao lower bound (CRLB) \cite{kay1993fundamentals,ucinski2004optimal}. And the CRLB is composed of the measurement condition and the geometric configuration. Specifically, the measurement condition includes the wireless propagation environment and the measurement capability of sensor cluster, and the geometric configuration refers to the measuring positional relationship between the sensor cluster and the source. Therefore, in the RSSD-based source localization system by drone swarm, the geometric configuration, namely measuring positions of UAVs, deserves to be optimized \cite{wang2018performance,cheng2022optimal}.

In  decades,  the geometric configuration optimization of source localization has been researched mainly following three optimization criterions, in which CRLB and Fisher information matrix (FIM) are used as the evaluation standards. These optimization criterions include D-optimality criterion (maximizing the determinant of FIM), A-optimality criterion (minimizing the trace of CRLB) and E-optimality criterion (minimizing the maximum eigenvalue of CRLB). Based on the above criterions, the geometric configuration of RSS-based localization system was optimized in \cite{zhao2013optimal,xu2019optimal,sahu2022optimal}.  A closed-form optimal geometric configuration for equal-weight cases where measurement capabilities of different sensor are identical was presented in \cite{zhao2013optimal}.
Besides, several numerical optimization algorithms were proposed  based on frame theory \cite{zhao2013optimal}, resistive network method \cite{xu2019optimal} and ADMM \cite{sahu2022optimal}. Compared to the RSS-based localization, the geometric configuration optimization of RSSD-based localization is more complex due to the unknown transmit power. And it has not been investigated until recent years \cite{li2022optimal,heydari2020optimal,heydari2020sensor}. In \cite{li2022optimal}, a closed-form optimal geometric configuration was derived for the scenario where three sensors were deployed. In \cite{heydari2020optimal}, only  the equal-weight scenario was considered, the theoretical analysis showed that when the geometric configuration reaches the optimum, the sensors are uniformly distributed on a circle centred on the source. In \cite{heydari2020sensor}, two suboptimal geometric configuration optimization algorithms were proposed, one was based on gradient descent and the other was based on sequential optimization.

It is worth mentioning that works in \cite{li2022optimal,heydari2020optimal,heydari2020sensor}  have not considered the measuring region constraint. However, due to limitations such as terrain, detection of enemy and communication range, the measuring positions of sensors are often constrained in a practical localization system by UAVs. Certainly, this constraint makes the optimization problem more cumbersome. In recent years, scholars have studied this kind of optimization problem with different region constraints for several types of localization. In \cite{cheng2022optimal}, the geometric configuration optimization problem with  distance constraints between sensors and the source for RSS-based localization was analyzed, and several numerical algorithms based on frame theory were proposed to find the optimal solutions. In \cite{sadeghi2020optimal},  any closed measuring region away from the source was considered, and an analytical geometric configuration for AoA localization was derived. For TDoA localization with any measuring region with spread angle less than $\frac{\pi}{2}$, an analytical geometric configuration was derived in \cite{sadeghi2021optimal}. In \cite{liang2016constrained}, the geometric configuration optimization problem for mixed range/angle/RSS localization was studied while the measuring region was a bow-shaped arch centered on the source, the max values of objective function versus the spread angle of this arch were given out. Moreover it was concluded that the best measurement is reachable when the spread angle is greater than $\frac{3\pi}{4}$.  However, it is hard to extend these works in  \cite{cheng2022optimal,sadeghi2020optimal,sadeghi2021optimal,liang2016constrained} to the RSSD-based localization system.

To overcome the shortcoming of existing works, in this paper, the geometric configuration optimization of RSSD-based source localization system by drone swarm is studied, and a general constrained measuring region is considered. Specifically, the measuring positions of UAVs are optimized with the spread angle constraint of drone swarm with respect to the source. The major contributions of this paper can be summarized as follows:
\begin{enumerate}
\item For the RSSD-based source localization system by drone swarm, considering constrained spread angle of drone swarm, the geometric configuration optimization problem is formulated following the D-optimization criterion. Then the multi-scalar optimization problem is equivalently transformed into a multi-vector optimization problem  and a single-matrix optimization problem. These derived problem forms provide valuable insight for the constrained geometric configuration optimization.
\item An ADMM-based constrained geometric configuration optimization algorithm is developed. Specifically, the ADMM framework is introduced to solve the derived single-matrix optimization problem. To solve the two subproblems in ADMM iterations,  a closed-form optimal solution and a MM-based optimal algorithm are proposed respectively. Although the proposed algorithm is designed for the RSSD-based localization, it can be also extended easily to other types of localization including RSS, DoA and TDoA.
\item Extensive simulations are conducted for various testing scenarios. The simulation results verify the effectiveness of the proposed algorithm on improving the upper limit of localization precision.  Furthermore, considering an initial estimation error in practical scenarios, the practicality of the proposed algorithm on refining the initial estimation is demonstrated.
\end{enumerate}

The rest of this paper is organized as follows. The measuring model is presented and the constrained geometric configuration optimization problem is formulated in Section II. In Section III, the optimization problem is transformed into a vector form and a matrix form. Then an ADMM-based algorithm is proposed in Section IV. In Section V, numerical results are presented. Finally, conclusions are drawn in Section VI.

\emph{Notations:} Boldface lower case and boldface upper case letters denote vectors and matrices, respectively. Sign $(\cdot)^{T}$  denotes the transpose operation. $|| \cdot ||$ denotes the Frobenius norm of a matrix or the L2 norm of a vector. $\mathbb{E}\{\cdot\}$ represents the expectation operation. $\mathrm{Tr}(\cdot)$ represents the trace of a matrix. $| \cdot |$ represents the determinant of a matrix. $\mathrm{Re}$ represents the real part of a number.

\section{Measurement Model and Problem Form}
\begin{figure}[!ht]
  \centering
  \includegraphics[width=0.45\textwidth]{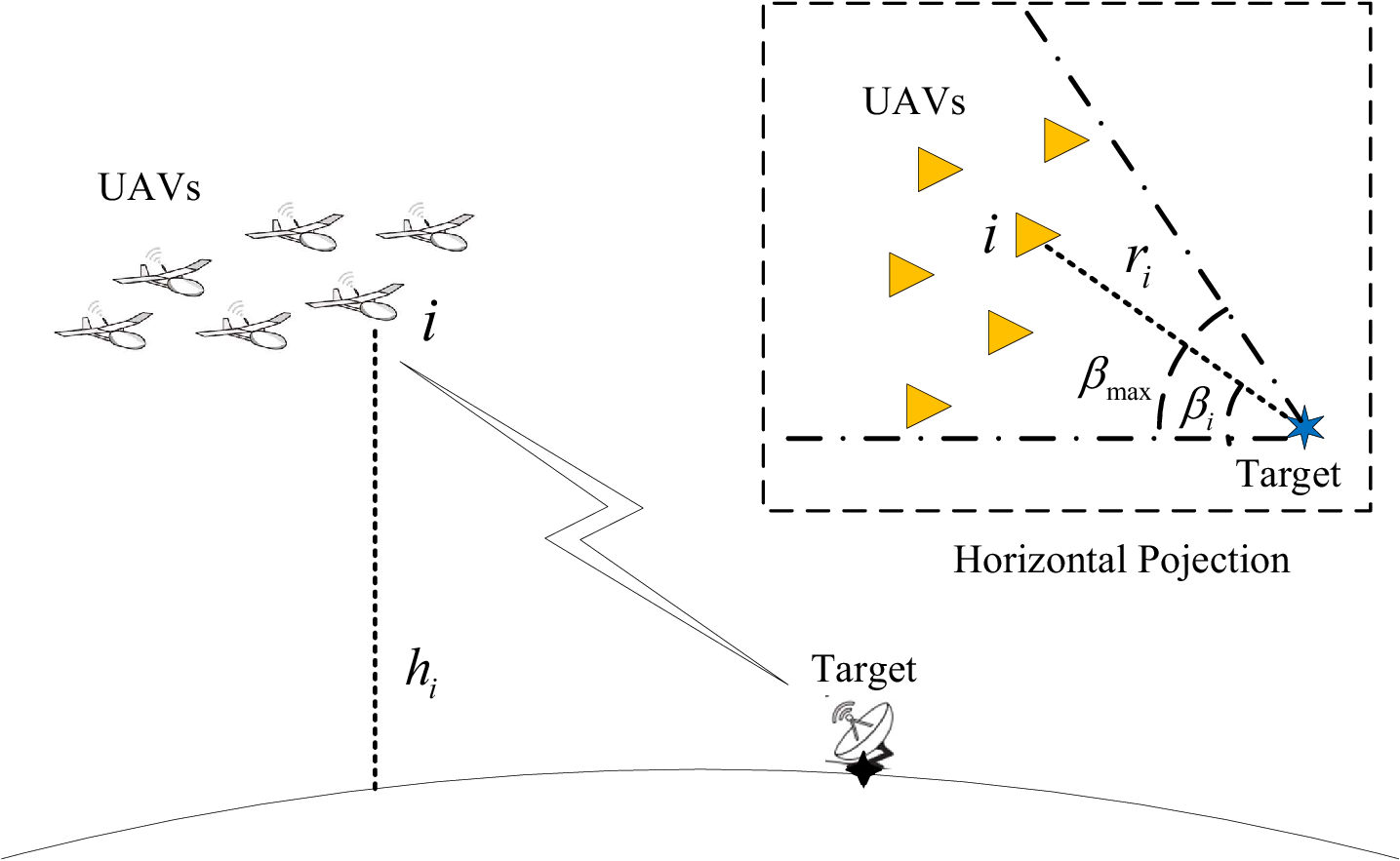}\\
  \caption{RSSD-based Source Localization via UAVs with Spread Angle Constraint.}\label{sys}
\end{figure}

\subsection{Measurement Model}
Consider a RSSD-based source localization system by UAVs as shown in Fig.~\ref{sys}. The source is located on the ground and emits wireless signal with unknown transmit power. In the localization task, three phases are involved. Initially, UAVs hovering at different positions receive wireless signals from the source and measure their strength using onboard sensors. Subsequently, each UAV transmits its locally stored RSS data, along with its measuring position, to either a central UAV within the swarm or to a ground control station. It is important to note that, due to the global positioning system (GPS) equipped on each UAV, the measuring position is accurately known. Finally, based on the aggregated data from the UAVs, a position estimation algorithm is applied to determine the target's location at the central UAV or ground control station.

Assume $N$ UAVs are employed. The position of target is expressed as $\mathbf{s}=\begin{bmatrix} x & y & 0 \end{bmatrix}$.
In the measurement phase, the position of the $i$-th UAV is expressed as $\mathbf{u}_{i}=\begin{bmatrix} x_{i} & y_{i}  & z_{i} \end{bmatrix}$, and the distance between  the $i$-th UAV and the target, denoted as $d_{i}$, is given by
\begin{align}
d_{i}=\sqrt{(x-x_{i})^2+(y-y_{i})^2+z_{i}^2},~i=1,2,\cdots,N.
\end{align}
According to the traditional path loss model of wireless  propagation (as shown in \cite{weiss2003accuracy}), the RSS of the $i$-th UAV is expressed as
\begin{align}\label{RSS}
P_{i}=\underbrace{P_{0}-10\gamma\lg(d_{i})}_{f_{i}(\bm{\theta})}+\eta_{i},~\forall i,
\end{align}
where $P_{0}$ represents the referenced power (in dB form) of source, $\gamma$ is the path loss exponent (PLE),  $\eta_{i}$ represents the measurement noise following the zero-mean Gaussian distribution with standard deviation $\sigma_{i}$, and $\bm{\theta}\triangleq\begin{bmatrix} P_{0} & \mathbf{s}^{T} \end{bmatrix}^{T}$. The term $\bm{\theta}$ is used to describe all unknown variables in this measurement.

Based on the Eq. (\ref{RSS}),  all measurements by drone swarm can be represented in a vector form as follows.
\begin{align}\label{RSS_vector}
\mathbf{P}=\mathbf{f}(\bm{\theta})+\bm{\eta},
\end{align}
where $\mathbf{P}\in \mathbb{R}^{N \times1}$ is the vector form of $\{P_{i}\}_{i=1}^{N}$, $\mathbf{f}(\mathbf{s})\in \mathbb{R}^{N\times1}$ is the vector form of  $\{f_{i}\}_{i=1}^{N}$, and  $\bm{\eta}\in \mathbb{R}^{N \times1}$ is the vector form of the $\{\eta_{i}\}_{i=1}^{N}$. Accordingly, the probability density function (PDF) of $\bm{\theta}$ is given by
\begin{align}\label{7BigML}
Q(\bm{\theta})= \frac{1}{(2\pi)^{\frac{N}{2}}|\mathbf{N}|^{\frac{1}{2}}}\exp\{-\frac{1}{2}(\mathbf{r}-\mathbf{f}(\bm{\theta}))^{T}\mathbf{N}^{-1}(\mathbf{r}-\mathbf{f}(\bm{\theta}))\}.
\end{align}

\subsection{Problem of Geometric Configuration Optimization}
In the estimation phase, following maximum likelihood estimate (MLE), numerical methods can be used to infer the position of source based on the Eq. (\ref{7BigML}). We denote an unbiased estimation by $\bm{\hat{\theta}}$, the covariance matrix of its error satisfies the following well-known inequality \cite{kay1993fundamentals}.
\begin{align}\label{errorandcandf}
\mathbb{E}\left\{(\bm{\hat{\theta}}-\bm{\theta}) (\bm{\hat{\theta}}-\bm{\theta})^{T}\right\}\succeq \mathbf{C}=\mathbf{F}^{-1}.
\end{align}
where $\mathbf{C}$ is the CRLB matrix, and $\mathbf{F}$ is the FIM \cite{kay1993fundamentals} given by
\begin{align}\label{FIM0}
\mathbf{F}=\mathbb{E}_{\bm{\eta}}\{\frac{\partial\ln{Q(\bm{\theta})}}{{\partial\bm{\theta}}} \frac{\partial^{T}\ln{Q(\bm{\theta})}}{{\partial\bm{\theta}}}\}.
\end{align}

Eq. (\ref{errorandcandf}) implies that $\mathbf{C}$ and $\mathbf{F}$ can be used to guide the measurement of drone swarm.
It is well known that the measuring positions of UAVs with respect to the source, namely geometric configuration, are important components of $\mathbf{C}$ or $\mathbf{F}$ \cite{dogancay2012uav,cheng2022optimal}. Therefore these measuring positions are optimized in this paper accordingly to improve the measurement quality as well as the final estimation accuracy. Moreover, the D-optimization criterion, i.e., maximizing the determinant of FIM, is applied to guide this optimization. Note that maximizing the determinant of FIM  is equivalent to maximizing the amount of information contained in the noisy measurements and minimizing the volume of the uncertainty ellipsoid \cite{ucinski2004optimal,liang2016constrained}.

To make the geometric configuration optimization concise, we reexpress the measuring positions of UAVs with respect to the source. The horizontal distance between the $i$-th UAV and the source is denoted as $r_{i}$. The vertical distance between the $i$-th UAV and the source, which is also the flying hight of the $i$-th UAV, is denoted as $h_{i}$. The horizontal angle of the $i$-th UAV with respect to the source is denoted as $\beta_{i}$, i.e., $\tan(\beta_{i})=\frac{x_{i}-x}{y_{i}-y}$. Note that in the three-dimensional space, $\mathbf{u}_{i}$ is determined uniquely by $r_{i}$, $h_{i}$ and $\beta_{i}$. Assume the distances between each UAV and the source are fixed \footnote{The optimization of distances between each UAV and the source is straightforward, and  independent of the optimization of angles  \cite{cheng2022optimal}.}. But the horizontal angles of UAVs are variables to be optimized. Moreover, as shown in Fig.~\ref{sys}, a general constraint is considered where the spread horizontal angle of drone swarm centered on the source is limited, i.e.,  $0\leq\beta_{i}\leq \beta_{max},~\forall i$.

Following the D-optimization criterion, the geometric configuration optimization problem for the RSSD-based localization system by drone swarm under the above measuring constraint  is formulated as
\begin{align}\label{pl}
\mathrm{(P1)}:&\max_{\{\beta_{i}\}_{i=1}^{N}}~~~~~~~~~~|\mathbf{F}|    \nonumber\\
&\text{s. t.}~~~~~0\leq\beta_{i}\leq \beta_{max},~\forall i.
\end{align}


\section{Problem Transformation and Analysis}\label{SeIII}
In this section, the specific expression of FIM is derived. Then the problem $\mathrm{(P1)}$ is transformed into a multi-vector optimization problem by replacing the horizontal angle with a unit directional vector. Based on this, a theorem about equivalent geometric configurations is presented.  Finally, the problem is  transformed into a concise single-matrix optimization problem by improving the dimension of vectorial variables.

From Eq. (\ref{FIM0}), the elements of FIM are given by
\begin{align}
\mathbf{F}_{i,j}=\mathbb{E}_{\bm{\eta}}\{\frac{\partial\ln{Q(\bm{\theta})}}{{\partial\bm{\theta}_{i}}} \frac{\partial^{T}\ln{Q(\bm{\theta})}}{{\partial\bm{\theta}_{j}}}\},~~~i\leq3,~~j\leq3.
\end{align}
Thus, we get
\begin{subequations}
\begin{align}
&\mathbf{F}_{1,1}=\mathbb{E}_{\bm{\eta}}\{\frac{\partial\ln{Q(\bm{\theta})}}{\partial P_{0}} \frac{\partial^{T}\ln{Q(\bm{\theta})}}{\partial P_{0}}\}=\mathbf{J}^{T}\mathbf{N}^{-1}\mathbf{J},\\
&\mathbf{F}_{2,2}=\mathbb{E}_{\bm{\eta}}\{\frac{\partial\ln{Q(\bm{\theta})}}{\partial x} \frac{\partial^{T}\ln{Q(\bm{\theta})}}{\partial x}\}=\mathbf{a}_{x}^{T}\mathbf{N}^{-1}\mathbf{a}_{x}, \\
&\mathbf{F}_{3,3}=\mathbb{E}_{\bm{\eta}}\{\frac{\partial\ln{Q(\bm{\theta})}}{\partial y} \frac{\partial^{T}\ln{Q(\bm{\theta})}}{\partial y}\}=\mathbf{a}_{y}^{T}\mathbf{N}^{-1}\mathbf{a}_{y},\\
&\mathbf{F}_{1,2}=\mathbf{F}_{2,1}=\mathbb{E}_{\bm{\eta}}\{\frac{\partial\ln{Q(\bm{\theta})}}{\partial P_{0}} \frac{\partial^{T}\ln{Q(\bm{\theta})}}{\partial x}\}=\mathbf{J}^{T}\mathbf{N}^{-1}\mathbf{a}_{x},\\
&\mathbf{F}_{1,3}=\mathbf{F}_{3,1}=\mathbb{E}_{\bm{\eta}}\{\frac{\partial\ln{Q(\bm{\theta})}}{\partial P_{0}} \frac{\partial^{T}\ln{Q(\bm{\theta})}}{\partial y}\}=\mathbf{J}^{T}\mathbf{N}^{-1}\mathbf{a}_{y},\\
&\mathbf{F}_{2,3}=\mathbf{F}_{3,2}=\mathbb{E}_{\bm{\eta}}\{\frac{\partial\ln{Q(\bm{\theta})}}{\partial x} \frac{\partial^{T}\ln{Q(\bm{\theta})}}{\partial y}\}=\mathbf{a}_{x}^{T}\mathbf{N}^{-1}\mathbf{a}_{y},
\end{align}
\end{subequations}
where
\begin{subequations}
\begin{align}
&\mathbf{J}=\left[1,1,\cdots,1  \right], \\
&\mathbf{a}_{x}= \begin{bmatrix}  \frac{10\gamma}{\ln{10}} \frac{x_{1}-x}{d_{1}^2} &\frac{10\gamma}{\ln{10}} \frac{x_{2}-x}{d_{2}^2} & \cdots & \frac{10\gamma}{\ln{10}} \frac{x_{N}-x}{d_{N}^2}\end{bmatrix},   \\
&\mathbf{a}_{y}=\begin{bmatrix} \frac{10\gamma}{\ln{10}} \frac{y_{1}-y}{d_{1}^2} &\frac{10\gamma}{\ln{10}} \frac{y_{2}-y}{d_{2}^2} &\cdots& \frac{10\gamma}{\ln{10}} \frac{y_{N}-y}{d_{N}^2}\end{bmatrix}.
\end{align}
\end{subequations}

Thus, the FIM can be partitioned into four submatrices as
\begin{subequations}\label{fim1}
\begin{align}
&\mathbf{F}=\begin{bmatrix} a  &  \mathbf{b}^{T}  \\ \mathbf{b} & \mathbf{H}  \end{bmatrix},  \\
&a=\mathbf{J}^{T}\mathbf{N}^{-1}\mathbf{J},  \\
&\mathbf{b}=\begin{bmatrix}\mathbf{J}^{T}\mathbf{N}^{-1}\mathbf{a}_{x} & \mathbf{J}^{T}\mathbf{N}^{-1}\mathbf{a}_{y} \end{bmatrix},   \\
&\mathbf{H}=\begin{bmatrix} \mathbf{a}_{x}^{T}\mathbf{N}^{-1}\mathbf{a}_{x}  &  \mathbf{a}_{x}^{T}\mathbf{N}^{-1}\mathbf{a}_{y}  \\ \mathbf{a}_{x}^{T}\mathbf{N}^{-1}\mathbf{a}_{y}  & \mathbf{a}_{y}^{T}\mathbf{N}^{-1}\mathbf{a}_{y} \end{bmatrix}.
\end{align}
\end{subequations}

To make the expression of FIM more intuitive and analyzable, we define the following symbols.
\begin{subequations}\label{suofang}
\begin{align}
\mathbf{W}\triangleq \mathrm{tr}^{-1}(\mathbf{N}^{-1})\mathbf{N} ^{-1}, \\
\overline{\sigma^{-2}}\triangleq \frac{1}{N}\sum_{i=1}^{N} \sigma_{i}^{-2}.
\end{align}
\end{subequations}
Based on Eq. (\ref{fim1}) and Eq. (\ref{suofang}), the objective function of problem $\mathrm{(P1)}$, i.e., the determinant of FIM is given by
\begin{align}\label{fimhanglie0}
|\mathbf{F}|=|a| |\mathbf{H}- \mathbf{b} a^{-1} \mathbf{b}^{T}|=N^2\overline{\sigma^{-2}}^{2}|\mathbf{T}^{'}|,
\end{align}
where,
\begin{align}\label{expressT}
&\mathbf{T}^{'}\nonumber\\
&=\begin{bmatrix} \mathbf{a}_{x}^{T}\mathbf{W}\mathbf{a}_{x}-(\mathbf{a}_{x}^{T}\mathbf{W}\mathbf{J})^{2}&  \mathbf{a}_{x}^{T}\mathbf{W}\mathbf{a}_{y}- \mathbf{a}_{x}^{T}\mathbf{W} \mathbf{J}\mathbf{a}_{y}^{T}\mathbf{W}\mathbf{J} \\ \mathbf{a}_{x}^{T}\mathbf{W}\mathbf{a}_{y}- \mathbf{a}_{x}^{T}\mathbf{W} \mathbf{J}\mathbf{a}_{y}^{T}\mathbf{W}\mathbf{J} &  \mathbf{a}_{y}^{T}\mathbf{W}\mathbf{a}_{y}- (\mathbf{a}_{y}^{T}\mathbf{W}\mathbf{J})^2  \end{bmatrix}\nonumber\\
&=\begin{bmatrix} \mathbf{a}_{x}^{T}(\mathbf{W}-\mathbf{W}\mathbf{J}(\mathbf{W}\mathbf{J})^{T})\mathbf{a}_{x}&  \mathbf{a}_{x}^{T}(\mathbf{W}-\mathbf{W}\mathbf{J}(\mathbf{W}\mathbf{J})^{T})\mathbf{a}_{y}
\\ \mathbf{a}_{x}^{T}(\mathbf{W}-\mathbf{W}\mathbf{J}(\mathbf{W}\mathbf{J})^{T})\mathbf{a}_{y}&  \mathbf{a}_{y}^{T}(\mathbf{W}-\mathbf{W}\mathbf{J}(\mathbf{W}\mathbf{J})^{T})\mathbf{a}_{y} \end{bmatrix}\nonumber\\
&=\begin{bmatrix} \mathbf{a}_{x}^{T} \\ \mathbf{a}_{y}^{T} \end{bmatrix} (\mathbf{W}-\mathbf{W}\mathbf{J}(\mathbf{W}\mathbf{J})^{T})\begin{bmatrix} \mathbf{a}_{x} & \mathbf{a}_{y}\end{bmatrix}.
\end{align}

Considering the relative position relationship between UAVs and the source as described in Section II, we have
\begin{align}\label{defg}
\begin{bmatrix} \mathbf{a}_{x}^{T} \\ \mathbf{a}_{y}^{T} \end{bmatrix}_{i}
&=\begin{bmatrix}  \frac{10\gamma}{\ln{10}} \frac{x_{i}-x}{d_{i}^2} &  \frac{10\gamma}{\ln{10}} \frac{y_{i}-y}{d_{i}^2} \end{bmatrix}\nonumber\\
&=\begin{bmatrix}  \frac{10\gamma}{\ln{10}} \frac{r_{i}\cos{\beta_{i}}}{d_{i}^2} &  \frac{10\gamma}{\ln{10}} \frac{r_{i}\sin{\beta_{i}}}{d_{i}^2} \end{bmatrix} \nonumber\\
&=\frac{10\gamma}{\ln{10}} \frac{r_{i}}{d_{i}^2} \begin{bmatrix}  \cos{\beta_{i}} & \sin{\beta_{i}}\end{bmatrix}=\frac{10\gamma}{\ln{10}} \frac{r_{i}}{d_{i}^2}\mathbf{g}_{i},~\forall i,
\end{align}
where $\begin{bmatrix} \mathbf{a}_{x}^{T} \\ \mathbf{a}_{y}^{T} \end{bmatrix}_{i}$ is the $i$-th column of $\begin{bmatrix} \mathbf{a}_{x}^{T} \\ \mathbf{a}_{y}^{T} \end{bmatrix}$, and $\mathbf{g}_{i}\triangleq\begin{bmatrix} \cos{\beta_{i}} & \sin{\beta_{i}}\end{bmatrix}$.
This defined vector is a unit directional vector, and accounts for the influencing factor to the measuring information amount caused by the horizontal angle of the $i$-th UAV. Moreover, it is easy to verify that there exists a one-one onto mapping between $\mathbf{g}_{i}$ and  $\beta_{i}$, which indicates that it is reasonable to replace $\{\beta_{i}\}_{i=1}^{N}$  with $\{\mathbf{g}_{i}\}_{i=1}^{N}$ as the optimization variables in problem $\mathrm{(P1)}$.

Substituting  (\ref{defg}) into (\ref{expressT}), we have
\begin{align}\label{deT}
\mathbf{T}^{'}=&\left(\frac{10\gamma}{\ln{10}}\right)^2\nonumber\\
&\left(\underbrace{\sum_{i=1}^{N}w_{i}\frac{r_{i}^2}{d_{i}^4} \mathbf{g}_{i}\mathbf{g}_{i}^{T}-
(\sum_{i=1}^{N}w_{i}\frac{r_{i}}{d_{i}^2} \mathbf{g}_{i}) (\sum_{i=1}^{N}w_{i}\frac{r_{i}}{d_{i}^2} \mathbf{g}_{i})^{T}}_{\mathbf{T}}\right).
\end{align}
And the determinant of FIM becomes
\begin{align}\label{fimhanglie0}
|\mathbf{F}|=\left(\frac{10\gamma}{\ln{10}}\right)^4N^2\overline{\sigma^{-2}}^{2}|\mathbf{T}|.
\end{align}

From above, the objective function of problem $\mathrm{(P1)}$ can be replaced by $|\mathbf{T}|$. Before transforming the problem $\mathrm{(P1)}$, we give a theorem about the geometric configuration by exploiting the properties of $|\mathbf{T}|$ as follows.
\begin{theorem}\label{thm1}
Given arbitrary but fixed coefficients $\{r_{i}\}_{i=1}^{N}$ and $\{h_{i}\}_{i=1}^{N}$, if two geometric configurations can convert to each other only by  horizontal overall rotation or horizontal overall reflection with respect to the source, then the corresponding determinants of FIM are equal, namely the two geometric configurations are equivalent.
\end{theorem}
\begin{proof}
See Appendix \ref{appendixa}.
\end{proof}

Theorem \ref{thm1} gives equivalent transformations for an arbitrary geometric configuration. It shows that the solutions of problem    $\mathrm{(P1)}$ are not unique. And using these transformations, we can obtain other equivalent solutions from a known one.

To replace $\{\beta_{i}\}_{i=1}^{N}$  with $\{\mathbf{g}_{i}\}_{i=1}^{N}$ in problem $\mathrm{(P1)}$, the objective function can be replaced with $|\mathbf{T}|$ in Eq (\ref{deT}), but the constraints in  problem $\mathrm{(P1)}$ still need to be transformed. Based on Theorem \ref{thm1}, the transformation is given by the following theorem.

\begin{theorem}\label{7thm2}
In problem $\mathrm{(P1)}$, $0\leq\beta_{i}\leq \beta_{max}$ is equivalent to $\mathbf{g}_{i}\geq \mathbf{g}_{0}$, where
\begin{align}
 \mathbf{g}_{0}=\left\{
	\begin{aligned}
	\begin{bmatrix} \cos{\beta_{max}} & 0\end{bmatrix}, \quad \beta_{max}\leq \pi\\
	\begin{bmatrix} -1 & \cos{\frac{\beta_{max}}{2}}\end{bmatrix}, \quad \beta_{max}>\pi\\
	\end{aligned}
	\right
..
\end{align}
\end{theorem}

\begin{proof}
Consider $\beta_{max}\leq \pi$ at first. It shows that $\cos{\beta_{i}}\geq\cos{\beta_{max}}$ means $\beta_{i}<\beta_{max} \cup \beta_{i}>2\pi-\beta_{max}$, and $\sin{\beta_{i}}\geq 0$ means $\beta_{i}\leq \pi$. Therefore, when $\mathbf{g}_{i}\geq \mathbf{g}_{0}$, $\beta_{i}$ satisfies $0\leq\beta_{i}\leq \beta_{max}$. Moreover, it is obvious that when $0\leq\beta_{i}\leq \beta_{max}$,  $\mathbf{g}_{i}$ satisfies $\mathbf{g}_{i}\geq \mathbf{g}_{0}$. From above, $0\leq\beta_{i}\leq \beta_{max}$ is equal to $\mathbf{g}_{i}\geq \mathbf{g}_{0}$ when $\beta_{max}\leq \pi$.

Then consider $\beta_{max}>\pi$. $\sin{\beta_{i}}\geq\cos{\frac{\beta_{max}}{2}}$ means $0\leq\beta_{i}\leq \frac{\pi}{2}+\frac{\beta_{max}}{2} \cup \frac{5\pi}{2}-\frac{\beta_{max}}{2}\leq\beta_{i}\leq 2\pi$.  According to Theorem \ref{thm1}, $0\leq\beta_{i}\leq \frac{\pi}{2}+\frac{\beta_{max}}{2} \cup \frac{5\pi}{2}-\frac{\beta_{max}}{2}\leq\beta_{i}\leq 2\pi$ is equivalent to $0\leq \beta_{i} \leq \beta_{max}$ in problem $\mathrm{(P1)}$. Therefore when $\mathbf{g}_{i}\geq \mathbf{g}_{0}$, $\beta_{i}$ satisfies $0\leq\beta_{i}\leq \beta_{max}$. Moreover, it is obvious that when $0\leq\beta_{i}\leq \beta_{max}$, $\mathbf{g}_{i}$ satisfies $\mathbf{g}_{i}\geq \mathbf{g}_{0}$. From above, $0\leq\beta_{i}\leq \beta_{max}$ is equal to $\mathbf{g}_{i}\geq \mathbf{g}_{0}$ when $\beta_{max}>\pi$.
\end{proof}

According to Eq (\ref{deT}), Eq (\ref{fimhanglie0}) and Theorem \ref{7thm2}, problem $\mathrm{(P1)}$ is equivalently transformed into
\begin{align}
\mathrm{(P2)}:&\max_{\{\mathbf{g}_{i}\}_{i=1}^{N}}~~~~~~~~~~ \left|\mathbf{T}\right| \nonumber\\
&~~\text{s. t.}~~~~~~~~~~\mathbf{g}_{i}^{T}\mathbf{g}_{i}=1, \nonumber\\
&~~~~~~~~~~~~~~\mathbf{g}_{i}\geq \mathbf{g}_{0},~\forall i.
\end{align}

Although problem $\mathrm{(P2)}$ is more concise than problem  $\mathrm{(P1)}$, it is still intractable due to the coupling of variables associated with different UAVs. Therefore, we continue transform problem $\mathrm{(P2)}$. At first, we define a new matrix that collects all variables in problem $\mathrm{(P2)}$ as follows.
\begin{align}
\mathbf{G}\triangleq\begin{bmatrix} \mathbf{g}_{1}^{T} &  \mathbf{g}_{2}^{T} &  \cdots & \mathbf{g}_{N}^{T} \end{bmatrix}.
\end{align}
Besides, we define the following matrixes.
\begin{align}
\mathbf{D}\triangleq\mathrm{diag}\left(\begin{bmatrix} \frac{r_{1}}{d_{1}^2} & \frac{r_{2}}{d_{2}^2}& \cdots & \frac{r_{N}}{d_{N}^2} \end{bmatrix}\right),
\end{align}
\begin{align}
\mathbf{B}\triangleq\mathbf{W}-\mathbf{W}\mathbf{J}(\mathbf{W}\mathbf{J})^{T}.
\end{align}
Based on these definitions, $\mathbf{T}$ in the objective function of problem $\mathrm{(P2)}$ is rewritten as
\begin{align}
\mathbf{T}
&=\left(\frac{\ln{10}}{10\gamma}\right)^{4}\nonumber\\
&~\begin{bmatrix} \mathbf{a}_{x}^{T}(\mathbf{W}-\mathbf{W}\mathbf{J}(\mathbf{W}\mathbf{J})^{T})\mathbf{a}_{x}&  \mathbf{a}_{x}^{T}(\mathbf{W}-\mathbf{W}\mathbf{J}(\mathbf{W}\mathbf{J})^{T})\mathbf{a}_{y}
\\ \mathbf{a}_{x}^{T}(\mathbf{W}-\mathbf{W}\mathbf{J}(\mathbf{W}\mathbf{J})^{T})\mathbf{a}_{y}&  \mathbf{a}_{y}^{T}(\mathbf{W}-\mathbf{W}\mathbf{J}(\mathbf{W}\mathbf{J})^{T})\mathbf{a}_{y} \end{bmatrix}\nonumber\\
&=\left(\frac{\ln{10}}{10\gamma}\right)^{4} \begin{bmatrix} \mathbf{a}_{x}^{T} \\ \mathbf{a}_{y}^{T} \end{bmatrix} \mathbf{B}\begin{bmatrix} \mathbf{a}_{x} & \mathbf{a}_{y}\end{bmatrix} \nonumber\\
&=\mathbf{G}^{T} \mathbf{D}^{T} \mathbf{B} \mathbf{D} \mathbf{G}.
\end{align}

For an arbitrary vector $\bm{\xi}$ in $\mathbb{R}^{N}$, the following formula is satisfied.
 \begin{align}\label{psidef}
\bm{\xi}^{T}\mathbf{B}\bm{\xi}&=\bm{\xi}^{T}\mathbf{W}\bm{\xi}-(\bm{\xi}^{T}\mathbf{W}\mathbf{J})^2\nonumber\\
 &=\sum_{i=1}^{N} w(i) \bm{\xi}^{2}(i)- (\sum_{i=1}^{K} w(i)\bm{\xi}(i))^2  \nonumber\\
 &=\sum_{i=1}^{N}w(i)g(\bm{\xi}(i))-g(\sum_{i=1}^{K}w(i)\bm{\xi}(i))\geq0,
\end{align}
where, $g(x)=x^2,~\forall x$, $w(i)$ is the $i$-th diagonal element of $\mathbf{W}$. It is easy to verify that $w(i)>0,~\forall i$ and $\sum_{i=1}^{N}w(i)=1$. The last inequality in formula (\ref{psidef}) results from Jensen's inequality \cite{boyd2004convex}. According to formula (\ref{psidef}), $\mathbf{B}$ is a semi-positive definite matrix, thus $\mathbf{B}$ is feasible for the cholesky decomposition, i.e., $\mathbf{B}=\mathbf{B}^{\frac{1}{2}}\mathbf{B}^{\frac{1}{2}T}$. Defining a matrix $\mathbf{X}$ that $\mathbf{X}\triangleq\mathbf{B}^{\frac{1}{2}}\mathbf{D} \mathbf{G}$ yields $\mathbf{T}=\mathbf{X}^{T} \mathbf{X}$. Finally, problem $\mathrm{(P2)}$ is transformed into the following problem.
\begin{align}
\mathrm{(P3)}:&\min_{\mathbf{G}}~~ \ln{\left|(\mathbf{X}^{T} \mathbf{X})^{-1}\right|}  \nonumber\\
&\text{s. t.}~~~\mathbf{X}=\mathbf{B}^{\frac{1}{2}}\mathbf{D} \mathbf{G}, \nonumber\\
&~~~~~~~~~~~\mathbf{G}\in \mathcal{D},
\end{align}
where $\mathcal{D}=\{\mathbf{G}|\mathbf{g}_{i}^{T}\mathbf{g}_{i}=1,~\mathbf{g}_{i}\geq \mathbf{g}_{0},~\forall i\}$.

\section{ADMM-Based Constrained Geometric Configuration Optimization Algorithm}

\subsection{Framework of ADMM}\label{seADMM}
To tackle problem $\mathrm{(P3)}$, ADMM framework is introduced in this section. Note that ADMM is  powerful and widely-used in the matrix optimization problem. It takes the form of a decomposition-coordination procedure, in which the solutions to small local subproblems are coordinated to find a solution to a large global problem \cite{boyd2011distributed}. And it blends the benefits of dual decomposition and augmented Lagrangian methods.

At first, we introduce Lagrangian multiplier denoted as $\mathbf{V}$ and augmented Lagrangian parameter denoted as $\rho$. Then the augmented Lagrangian of problem $\mathrm{(P3)}$ is expressed as
\begin{align}
L_{\rho}(\mathbf{X},\mathbf{G},\mathbf{V})=&\ln{\left|(\mathbf{X}^{T} \mathbf{X})^{-1}\right|}
+\mathrm{Tr}\left(\mathbf{V}^{T} (\mathbf{B}^{\frac{1}{2}}\mathbf{D} \mathbf{G}- \mathbf{X}) \right) \nonumber\\
&+\frac{\rho}{2}\|\mathbf{B}^{\frac{1}{2}}\mathbf{D} \mathbf{G}- \mathbf{X}\|^{2}.
\end{align}
According to \cite{boyd2011distributed}, the ADMM update framework to solve problem $\mathrm{(P3)}$ is given by
\begin{subnumcases} {\label{eqADMM}}
\mathbf{X}^{k+1}=\mathop{\mathrm{argmin}}_{\mathbf{X}} L_{\rho}(\mathbf{X},\mathbf{G}^{k},\mathbf{V}^{k}), \label{step1}\\
\mathbf{G}^{k+1}=\mathop{\mathrm{argmin}}_{\mathbf{G}\in \mathcal{D}}  L_{\rho}(\mathbf{X}^{k+1},\mathbf{G},\mathbf{V}^{k}), \label{step2} \\
\mathbf{V}^{k+1}=\mathbf{V}^{k}+\rho(\mathbf{B}^{\frac{1}{2}}\mathbf{D} \mathbf{G}^{k+1}-\mathbf{X}^{k+1}), \label{step3}
\end{subnumcases}
where $k$ is the iteration  number. In this framework, the update of $\mathbf{V}$ is straightforward, but the updates of $\mathbf{X}$ and $\mathbf{G}$ are not. We will given these updates by finding the global optimal solutions of  subproblem (\ref{step1}) and  subproblem (\ref{step2}) in the following subsections.

\subsection{Update of $\mathbf{X}$}\label{7sematrix}
At the $(k+1)$-th iteration, when $\mathbf{G}^{k}$ and $\mathbf{V}^{k}$ are given,  $L_{\rho}(\mathbf{X},\mathbf{G}^{k},\mathbf{V}^{k})$
(simplified as $L_{\rho}^{k}$ in this subsection) in subproblem (\ref{step1})  can be written as
\begin{align}\label{subpx}
&L_{\rho}^{k}=\ln{\left|(\mathbf{X}^{T} \mathbf{X})^{-1}\right|}+\frac{\rho}{2}\mathrm{Tr}(\mathbf{X}^{T} \mathbf{X})-\mathrm{Tr}\left(\mathbf{J}^{k T}\mathbf{X}\right)+\beta^{k},
\end{align}
where $\mathbf{J}^{k}=\mathbf{V}^{k}+\rho\mathbf{B}^{\frac{1}{2}}\mathbf{D} \mathbf{G}^{k}$, $\beta^{k}=\mathbf{Tr}(\mathbf{V}^{kT}\mathbf{B}^{\frac{1}{2}}\mathbf{D} \mathbf{G}^{k})+\frac{\rho}{2}\mathrm{Tr}\left((\mathbf{B}^{\frac{1}{2}}\mathbf{D} \mathbf{G}^{k})^{T} (\mathbf{B}^{\frac{1}{2}}\mathbf{D} \mathbf{G}^{k})\right)$.

To address subproblem (\ref{step1}), singular value decomposition (SVD) is applied. Since $\mathbf{X}\in\mathbb{R}^{N\times2}$ and $\mathbf{J}^{k}\in\mathbb{R}^{N\times2}$, the SVDs are expressed as
\begin{align}\label{svd1}
\mathbf{X}=\mathbf{U} \begin{bmatrix} \mathbf{\Lambda} & \mathbf{0} \end{bmatrix} \mathbf{V},
\end{align}
\begin{align}\label{svd2}
\mathbf{J}^{k}=\mathbf{U}_{J}^{k}\begin{bmatrix} \mathbf{\Sigma}^{k} & \mathbf{0}
\end{bmatrix} \mathbf{V}_{J}^{k},
\end{align}
where $\mathbf{0}\in \mathbb{R}^{(N-2)\times 2}$, $\mathbf{\Lambda}=\mathrm{diag}\left(  \begin{bmatrix} \tau_{1} & \tau_{2} \end{bmatrix}\right)$ with $\tau_{1}\geq \tau_{2}\geq0$ and $\mathbf{\Sigma}^{k}=\mathrm{diag}\left( \begin{bmatrix} \sigma^{k}_{1} & \sigma^{k}_{2} \end{bmatrix}\right)$ with $\sigma^{k}_{1}\geq \sigma^{k}_{2}\geq0$. Note that $\{\tau_{i}\}_{i=1}^{2}$ are the non-zero singular values of $\mathbf{X}$ and $\{\sigma^{k}_{i}\}_{i=1}^{2}$ are the non-zero singular values of $\mathbf{J}^{k}$.
After substituting Eq. (\ref{svd1}) and Eq. (\ref{svd2}) into Eq. (\ref{subpx}), the objective function of subproblem  (\ref{step1}) becomes
\begin{align}
L_{\rho}^{k}=\ln{\left|(\mathbf{\Lambda}^{T} \mathbf{\Lambda})^{-1}\right|}
+\frac{\rho}{2}\mathrm{Tr}(\mathbf{\Lambda}^{T} \mathbf{\Lambda})-\mathrm{Tr}\left(\mathbf{J}^{k T}\mathbf{X}\right)+\beta^{k}.
\end{align}

According to Theorem 1 in \cite{xingdongnote}, which is a corollary of the Von Neumann's trace inequality \cite{horn2012matrix}, we have
\begin{align}
\mathrm{Tr}\left(\mathbf{J}^{k T}\mathbf{X}\right)\leq \sum_{i=1}^{2} |\mathrm{Re} \lambda_{i}(\tilde{\mathbf{T}})|\tau_{i}\sigma^{k}_{i},
\end{align}
where,
\begin{align}
\tilde{\mathbf{T}}=\mathbf{U}_{J}^{kT}\mathbf{U}\begin{bmatrix} \mathbf{V}\mathbf{V}_{J}^{kT} & \mathbf{0} \\ \mathbf{0} & \mathbf{I}_{N-2}\end{bmatrix},
\end{align}
where $\lambda_{i}(\tilde{\mathbf{T}})$ denotes the $i$-th eigenvalue of $\tilde{\mathbf{T}}$ with $\lambda_{i-1}(\tilde{\mathbf{T}})>\lambda_{i}(\tilde{\mathbf{T}})$.

Because $\tilde{\mathbf{T}}^{T}\tilde{\mathbf{T}}=\tilde{\mathbf{T}}\tilde{\mathbf{T}}^{T}=\mathbf{I}_{N}$, $\tilde{\mathbf{T}}$ is an
unitary matrix with $|\mathrm{Re} \lambda_{i}(\tilde{\mathbf{T}})|\leq 1$. Then we have,
\begin{align}\label{uptr}
\mathrm{Tr}\left(\mathbf{J}^{k T}\mathbf{X}\right)\leq \sum_{i=1}^{2} \tau_{i}\sigma^{k}_{i}.
\end{align}
It is easy to verify that when $\mathbf{U}=\mathbf{U}_{J}^{k}$ and $\mathbf{V}=\mathbf{V}_{J}^{k}$, the equation in formula (\ref{uptr}) holds, which indicates $L_{\rho}^{k}$ reaches the minimum. Under this condition,
\begin{align}
L_{\rho}^{k}=\ln{\left|(\mathbf{\Lambda}^{T} \mathbf{\Lambda})^{-1}\right|}
+\frac{\rho}{2}\mathrm{Tr}(\mathbf{\Lambda}^{T} \mathbf{\Lambda})-\mathrm{Tr}\left(\mathbf{\Sigma}^{k}\mathbf{\Lambda}\right)+\beta^{k}.
\end{align}

Thereby, $L_{\rho}^{k}$ has transformed as a function of $\{\lambda_{i}\}_{i=1}^{2}$. It is straightforward that the minimum of $L_{\rho}^{k}$ satisfies the following equation.
\begin{align}\label{tidi0}
\nabla L_{\rho}^{k}=-2\mathbf{\Lambda}^{-1} +\rho\mathbf{\Lambda}-\mathbf{\Sigma}^{k}=\mathbf{0}.
\end{align}
Due to the positive property of $\lambda_{i}$, it shows that  when $L_{\rho}^{k}$ reaches the minimum,
\begin{align}
\lambda_{i}=\frac{\sigma^{k}_{i}+\sqrt{(\sigma^{k}_{i})^{2}+8\rho}}{2\rho},~i=1,2.
\end{align}

From above, the update of $\mathbf{X}$ at the $(k+1)$-th iteration is given by
\begin{align}
\mathbf{X}^{k+1}=\mathbf{U}_{J}^{k}  \begin{bmatrix} \mathbf{\Lambda}^{k} & \mathbf{0} \end{bmatrix}   \mathbf{V}_{J}^{k},
\end{align}
where $\mathbf{\Lambda}^{k}=\mathrm{diag}\left(\begin{bmatrix} \frac{\sigma^{k}_{1}+\sqrt{(\sigma^{k}_{1})^{2}+8\rho}}{2\rho} & \frac{\sigma^{k}_{2}+\sqrt{(\sigma^{k}_{2})^{2}+8\rho}}{2\rho} \end{bmatrix} \right)$.

\subsection{Update of $\mathbf{G}$}\label{7semm}
At the $(k+1)$-th iteration, when $\mathbf{X}^{k+1}$ and $\mathbf{V}^{k}$ are given, $L_{\rho}(\mathbf{X}^{k+1},\mathbf{G},\mathbf{V}^{k})$ (simplified as $L_{\rho}^{k}$ in this subsection) in subproblem (\ref{step2}) can be written as
\begin{align}\label{lossG0}
&L_{\rho}^{k}\nonumber\\
&=\ln{\left|(\mathbf{X}^{(k+1)T} \mathbf{X}^{k+1})^{-1}\right|}
+\mathrm{Tr}\left(\mathbf{V}^{kT} (\mathbf{B}^{\frac{1}{2}}\mathbf{D} \mathbf{G}- \mathbf{X}^{k+1}) \right)\nonumber\\
&~~~+\frac{\rho}{2}\|\mathbf{B}^{\frac{1}{2}}\mathbf{D} \mathbf{G}- \mathbf{X}^{k+1}\|^{2}\nonumber\\
&=\frac{\rho}{2}\mathrm{Tr}\left(\mathbf{G}^{T}(\mathbf{B}^{\frac{1}{2}}\mathbf{D})^{T}\mathbf{B}^{\frac{1}{2}}\mathbf{D}\mathbf{G}  \right)+\mathrm{Tr}\left(\mathbf{C}^{kT}\mathbf{B}^{\frac{1}{2}}\mathbf{D}\mathbf{G}\right)+\alpha^{k},
\end{align}
where $\mathbf{C}^{k}=\mathbf{V}^{k}-\rho\mathbf{X}^{k+1}$, $\alpha^{k}=\ln{\left|(\mathbf{X}^{(k+1)T} \mathbf{X}^{k+1})^{-1}\right|}+\mathrm{Tr}\left(\frac{\rho}{2}\mathbf{X}^{(k+1)T}\mathbf{X}^{k+1}-\mathbf{V}^{kT}\mathbf{X}^{k+1} \right)$.
Due to the constraint in subproblem (\ref{step2}), it is difficult to derive any closed-form solution. To solve this non-trivial  problem, majorization-minimization (MM) technique is used here. Note that MM is a useful technique to solve complicated optimization problems appeared in signal processing and communications \cite{sun2016majorization,cheng2021communication}.

Make the following definitions: $\mathbf{M}\triangleq (\mathbf{B}^{\frac{1}{2}}\mathbf{D})^{T}\mathbf{B}^{\frac{1}{2}}\mathbf{D}$, $\mathbf{\tilde{M}}\triangleq \mathbf{M}-\lambda_{m}(\mathbf{M})\mathbf{I}\preceq \mathbf{0}$. Then $L_{\rho}^{k}$ can be written as
\begin{align}
L_{\rho}^{k}&=\frac{\rho}{2}\mathrm{Tr}\left(\mathbf{G}^{T}\mathbf{\tilde{M}}\mathbf{G}  \right) \nonumber\\
&~~~+\frac{\rho}{2}\lambda_{m}(\mathbf{M}) \mathrm{Tr}\left(\mathbf{G}^{T}\mathbf{G}\right)
+\mathrm{Tr}\left(\mathbf{C}^{kT}\mathbf{B}^{\frac{1}{2}}\mathbf{D}\mathbf{G}\right)+\alpha^{k}\nonumber\\
&=\frac{\rho}{2}\mathrm{Tr}\left(\mathbf{G}^{T}\mathbf{\tilde{M}}\mathbf{G}  \right)+\mathrm{Tr}\left(\mathbf{C}^{kT}\mathbf{B}^{\frac{1}{2}}\mathbf{D}\mathbf{G}\right)+\mu^{k},
\end{align}
where $\mu^{k}=\alpha^{k}+\frac{\rho}{2}N\lambda_{m}(\mathbf{M})$. The second equation results from $\mathrm{Tr}\left(\mathbf{G}^{T}\mathbf{G}\right)=N$.

At each iteration of MM, the global upper bound of the objective function of subproblem (\ref{step2}) should be constructed. To achieve it, we first construct a bound of $\mathrm{Tr}\left(\mathbf{G}^{T}\mathbf{\tilde{M}}\mathbf{G}  \right)$, then extend to $L_{\rho}^{k}$. It is easy to verify $\mathrm{Tr}\left(\mathbf{G}^{T}\mathbf{\tilde{M}}\mathbf{G}  \right)$ is a concave function with respect to $\mathbf{G}$. For any given $\mathbf{G}^{k}_{t}$, there must exist an upper bound of $\mathrm{Tr}\left(\mathbf{G}^{T}\mathbf{\tilde{M}}\mathbf{G}  \right)$, shown as follows.
\begin{align}
\mathrm{Tr}\left(\mathbf{G}^{T}\mathbf{\tilde{M}}\mathbf{G} \right)\leq 2\mathrm{Tr}\left(\mathbf{G}^{kT}_{t}\mathbf{\tilde{M}}\mathbf{G} \right)-\mathrm{Tr}\left(\mathbf{G}^{kT}_{t}\mathbf{\tilde{M}}\mathbf{G}^{k}_{t}\right).
\end{align}
Therefore,
\begin{align}
L_{\rho}^{k}\leq \hat{L}_{\rho,t}^{k}=\mathrm{Tr}\left(\mathbf{Q}^{kT}_{t}\mathbf{G}\right)+\upsilon_{t}^{k},
\end{align}
where $\mathbf{Q}^{kT}_{t}=\mathbf{C}^{kT}\mathbf{B}^{\frac{1}{2}}\mathbf{D}+\rho\mathbf{G}^{kT}_{t}\mathbf{\tilde{M}}$ and
$\upsilon_{t}^{k}=-\frac{\rho}{2}\mathrm{Tr}\left(\mathbf{G}^{kT}_{t}\mathbf{\tilde{M}}\mathbf{G}^{k}_{t}\right)+\mu^{k}$.

Denote the update of $\mathbf{G}^{k}$ at the $t$-th iteration of MM as $\mathbf{G}^{k}_{t}$. Then at the $(t+1)$-th iteration, the objective function needs to be minimized is just $\hat{L}_{\rho,t}^{k}$. Consider the constraint in subproblem (\ref{step2}), the $(t+1)$-th update of $\mathbf{G}^{k}$ is computed by solving the following problem.
\begin{align}
\mathbf{G}^{k}_{t+1}=\mathop{\mathrm{argmin}}_{\mathbf{G}\in \mathcal{D}} \mathrm{Tr}\left(\mathbf{Q}^{kT}_{t}\mathbf{G}\right)=\mathop{\mathrm{argmin}}_{\mathbf{G}\in \mathcal{D}} \sum_{i=1}^{N}\mathbf{g}_{i}^{T}\mathbf{q}_{i,t}^{k},
\end{align}
where $\mathbf{q}_{i,t}^{k}$ is the transpose of the $i$-th row of $\mathbf{Q}^{kT}_{t}$.

Let $\mathbf{g}_{i,t+1}^{k}$ denote the transpose of the $i$-th row of $\mathbf{G}^{k}_{t+1}$. After some basic algebraic operations, $\mathbf{g}_{i,t+1}^{k+1}$ is derived as follows.
\begin{align}\label{MMupdate}
&\mathbf{g}_{i,t+1}^{k}= \nonumber\\
&\left\{
	\begin{aligned}
	-\frac{\mathbf{q}_{i,t}^{k}}{\| \mathbf{q}_{i,t}^{k} \|},~~~~~~~~~~~~~~~~~~~~\quad -\frac{\mathbf{q}_{i,t}^{k}}{\| \mathbf{q}_{i,t}^{k} \|}\geq \mathbf{g}_{0}\\
	\mathop{\mathrm{argmin}}_{\mathbf{g}_{i}\in \{\mathbf{g}_{a},\mathbf{g}_{b}\}}\mathbf{g}_{i}^{T}\mathbf{q}_{i,t}^{k}, \quad \overline{-\frac{\mathbf{q}_{i,t}^{k}}{\| \mathbf{q}_{i}^{k,t} \|}\geq \mathbf{g}_{0}} \cap 0<\beta_{max}\leq \pi\\
	\mathop{\mathrm{argmin}}_{\mathbf{g}_{i}\in \{\mathbf{g}_{c},\mathbf{g}_{d}\}}\mathbf{g}_{i}^{T}\mathbf{q}_{i,t}^{k}, \quad \overline{-\frac{\mathbf{q}_{i,t}^{k}}{\| \mathbf{q}_{i,t}^{k} \|}\geq \mathbf{g}_{0} }\cap \pi<\beta_{max}\leq 2\pi \\
	\end{aligned}
	\right
.
,
\end{align}
where
\begin{subequations}
\begin{equation}
\mathbf{g}_{a}=\begin{bmatrix} \cos{\beta_{max}} & \sin{\beta_{max}}\end{bmatrix},
\end{equation}
\begin{equation}
\mathbf{g}_{b}=\begin{bmatrix} 1 & 0\end{bmatrix},
\end{equation}
\begin{equation}
\mathbf{g}_{c}=\begin{bmatrix} \cos{\frac{\pi+\beta_{max}}{2}} & \sin{\frac{\pi+\beta_{max}}{2}}\end{bmatrix},
\end{equation}
\begin{equation}
\mathbf{g}_{d}=\begin{bmatrix} \cos{\frac{5\pi-\beta_{max}}{2}} & \sin{\frac{5\pi-\beta_{max}}{2}}\end{bmatrix}.
\end{equation}
\end{subequations}

Assume the MM converges at the $T$-th iteration, then $\mathbf{G}^{k}_{T}$ is the final update of MM. Due to the property of MM, $\mathbf{G}^{k}_{T}$ is a local optimal solution of subproblem (\ref{step2}). But it is proofed that $\mathbf{G}^{k}_{T}$ is also the global optimal solution of subproblem (\ref{step2}). The details are shown in Appendix \ref{proofMM}.

From above, the update of $\mathbf{G}$ at the $(k+1)$-th iteration is given by
\begin{align}
\mathbf{G}^{k+1}=\mathbf{G}^{k}_{T}.
\end{align}

\subsection{Proposed Algorithm}
\begin{algorithm}[!ht]
\caption{The ADMM-Based Constrained Geometric Configuration Optimization Method (ADMM-CGCOM)}
\label{alg_1}
\begin{algorithmic}[1]
\STATE Input:  $\rho$, $\beta_{max}$, $\mathbf{B}$, $\mathbf{D}$, $\mathbf{\tilde{M}}$, $\mathbf{G}_{u}$, $\mathbf{g}_{0}$, $\mathbf{g}_{a}$, $\mathbf{g}_{b}$, $\mathbf{g}_{c}$, $\mathbf{g}_{d}$
\STATE Output: $\mathbf{G}$
\STATE Set $k=0$
\STATE Initialize $\mathbf{V}^{k}$, $\mathbf{G}^{k}=\mathbf{G}_{u}$, $\mathbf{X}^{k}=\mathbf{B}^{\frac{1}{2}}\mathbf{D}\mathbf{G}^{k}$
\REPEAT
\STATE $\mathbf{J}^{k}=\mathbf{V}^{k}+\rho\mathbf{B}^{\frac{1}{2}}\mathbf{D} \mathbf{G}^{k}$
\STATE $\mathbf{J}^{k}=\mathbf{U}_{J}^{k}\begin{bmatrix} \mathbf{\Sigma}^{k} & \mathbf{0}\end{bmatrix} \mathbf{V}_{J}^{k}$, where $\mathbf{\Sigma}^{k}=\mathrm{diag}\left( \begin{bmatrix} \sigma^{k}_{1} & \sigma^{k}_{2} \end{bmatrix}\right)$ ($\sigma^{k}_{1}\geq \sigma^{k}_{2}\geq0$)
\STATE $\mathbf{X}^{k+1}=\mathbf{U}_{J}^{k}  \begin{bmatrix} \mathbf{\Lambda}^{k} & \mathbf{0} \end{bmatrix}  \mathbf{V}_{J}^{k}$,
where $\mathbf{\Lambda}^{k}=\mathrm{diag}\left(\begin{bmatrix} \frac{\sigma^{k}_{1}+\sqrt{(\sigma^{k}_{1})^{2}+8\rho}}{2\rho} & \frac{\sigma^{k}_{2}+\sqrt{(\sigma^{k}_{2})^{2}+8\rho}}{2\rho} \end{bmatrix} \right)$
\STATE Set $t=0$
\STATE $\mathbf{G}_{t}^{k}=\mathbf{G}^{k}$
\REPEAT
\STATE $\mathbf{Q}^{kT}_{t}=(\mathbf{V}^{k}-\rho\mathbf{X}^{k+1})^{T}\mathbf{B}^{\frac{1}{2}}\mathbf{D}+\rho\mathbf{G}^{kT}_{t}\mathbf{\tilde{M}}$,
$\begin{bmatrix} \mathbf{q}_{1,t}^{kT} &  \mathbf{q}_{2,t}^{kT} &  \cdots & \mathbf{q}_{N,t}^{kT} \end{bmatrix}=\mathbf{Q}^{kT}_{t}$
\STATE $i=1$
\REPEAT
\IF{$-\frac{\mathbf{q}_{i,t}^{k}}{\| \mathbf{q}_{i,t}^{k} \|}\geq \mathbf{g}_{0}$}
\STATE $\mathbf{g}_{i,t+1}^{k}=-\frac{\mathbf{q}_{i,t}^{k}}{\| \mathbf{q}_{i,t}^{k}\|}$
\ELSE
\IF{$0<\beta_{max}\leq \pi$}
\STATE $\mathbf{g}_{i,t+1}^{k}=\mathop{\mathrm{argmin}}_{\mathbf{g}_{i}\in \{\mathbf{g}_{a},\mathbf{g}_{b}\}}\mathbf{g}_{i}^{T}\mathbf{q}_{i,t}^{k}$
\ELSE
\STATE $\mathbf{g}_{i,t+1}^{k}=\mathop{\mathrm{argmin}}_{\mathbf{g}_{i}\in \{\mathbf{g}_{c},\mathbf{g}_{d}\}}\mathbf{g}_{i}^{T}\mathbf{q}_{i,t}^{k}$
\ENDIF
\ENDIF
\STATE $i=i+1$
\UNTIL {$i=N$}
\STATE  Update $\mathbf{G}_{t+1}^{k}=\begin{bmatrix} \mathbf{g}_{1,t+1}^{kT} &  \mathbf{g}_{2,t+1}^{kT} &  \cdots & \mathbf{g}_{N,t+1}^{kT} \end{bmatrix}$
\STATE $t=t+1$
\UNTIL {convergence}
\STATE Update $\mathbf{G}^{k+1}=\mathbf{G}_{t+1}^{k}$
\STATE Update $\mathbf{V}^{k+1}=\mathbf{V}^{k}+\rho(\mathbf{B}^{\frac{1}{2}}\mathbf{D} \mathbf{G}^{k+1}-\mathbf{X}^{k+1})$
\STATE $k=k+1$
\UNTIL {convergence}
\end{algorithmic}
\end{algorithm}

So far, we have derived the update rules in ADMM framework. Since problem $\mathrm{(P3)}$ is non-convex, although convergent\footnote{The convergence of non-convex ADMM is still an open question. As illustrated in \cite{sahu2022optimal}, the approaches mentioned in \cite{hong2016convergence,wang2019global,liu2019linearized} can be adapted to prove the convergence of ADMM iterations to a KKT point of the respective optimal design problems. Moreover, in our simulation studies, the ADMM always convergence on the focused problem.}, the non-convex ADMM may not lead to a global optimal solution. To make the convergence of non-convex ADMM more effective, the initialization of ADMM should be meticulously designed. Under simulation studies, the uniform placement strategy is adopted for the initialization. Note that uniform placement strategy has been demonstrated to be optimal on some basic geometric configuration optimizations  \cite{zhao2013optimal,heydari2020optimal}. As for the focused problem, the horizontal angles of UAVs following this strategy are given by
\begin{align}\label{iniangle}
\beta_{u,i}=\frac{\beta_{max}}{N}i,~\forall i.
\end{align}
Let $\mathbf{G}_{u}$ denote the related $\mathbf{G}$ under (\ref{iniangle}). Finally, the proposed ADMM-based constrained geometric configuration optimization method (ADMM-CGCOM) is given by Algorithm \ref{alg_1}.

\begin{remark}
Although Algorithm \ref{alg_1} is designed for the RSSD-based localization as described in Section II, it can be extended  to other source localization technologies well. Specifically, for RSS-based localization \cite{xu2019optimal} in where the transmit power of source is known, Algorithm \ref{alg_1} still works by defining $B\triangleq \mathbf{W}$. Moreover, Algorithm \ref{alg_1} also works for TOA-based localization \cite{xu2021optimal} and TDOA-based localization \cite{isaacs2009optimal} by defining $\mathrm{B}$ and $\mathrm{D}$ accordingly. These extensions are straightforward, since general frameworks of these technologies on FIM have already been presented  \cite{sahu2022optimal,zhao2013optimal}, thus we omit here.
\end{remark}

\begin{remark}
Now, the main computational complexities of Algorithm \ref{alg_1} are analyzed in terms of floating-point operations. To generalize the result, let $K$ as the dimension of estimated position. The computational complexity for initialization, which primarily  depends on computing the square root and  matrix multiplication, is approximately  $\mathcal{O}(N^3+N^2K)$. For the outer loop, its computational complexity arises from computing the SVD and matrix multiplication, and is approximated as $\mathcal{O}(N^2K+NK^2+NK+N^2)$. For the inner loop, its computational complexity mainly involves the traversal and matrix multiplication, and is approximated as $\mathcal{O}(N^2K+NK)$. Therefore, the overall computational complexity of Algorithm \ref{alg_1} is $\mathcal{O}(N^3+I_{1}( N^2K+NK^2+I_{2}(N^2K)))$, where $I_{1}$  and $I_{2}$ represent the iterations of the outer and  inner loop, respectively.
\end{remark}

\begin{remark}
In some practical scenarios, both the UAV-target distance and UAV-target horizontal angle are need to be constrained and optimized. Thorough analysis in Appendix \ref{proofre3} shows that the optimization of horizontal angle and the optimization of distance are independent. Furthermore, regardless of the horizontal angular configuration, the optimal UAV-target distance of UAV remains unchanged for each UAV. These findings substantiate that Algorithm \ref{alg_1} is applicable and effective in these practical scenarios.
\end{remark}


\section{Simulations and Discussions}

\begin{figure}
   \subfloat[ \label{fig_s1}$\beta_{max}=120^{o}$.$\bar{k}_{MM}=3$.]{%
       \includegraphics[width=0.235\textwidth]{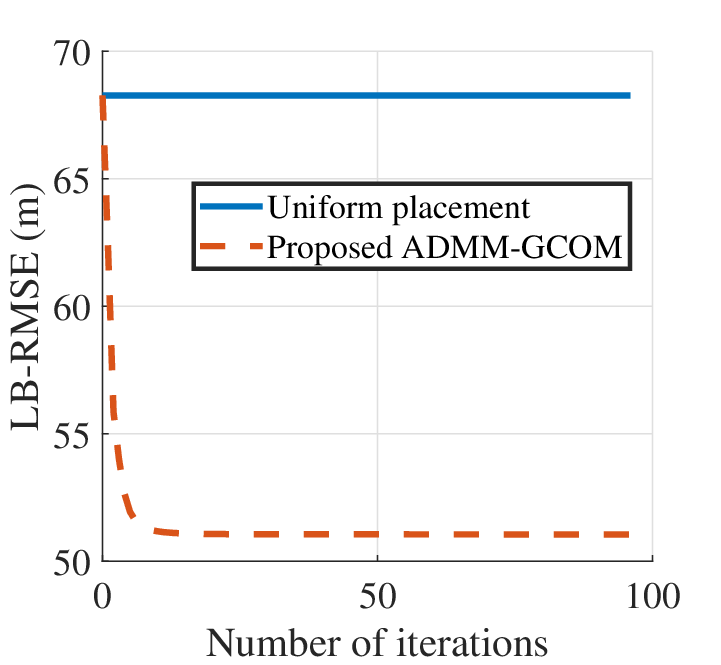}
     }
     \hfill
     \subfloat[\label{fig_s2}$\beta_{max}=200^{o}$.$\bar{k}_{MM}=4$.]{%
       \includegraphics[width=0.235\textwidth]{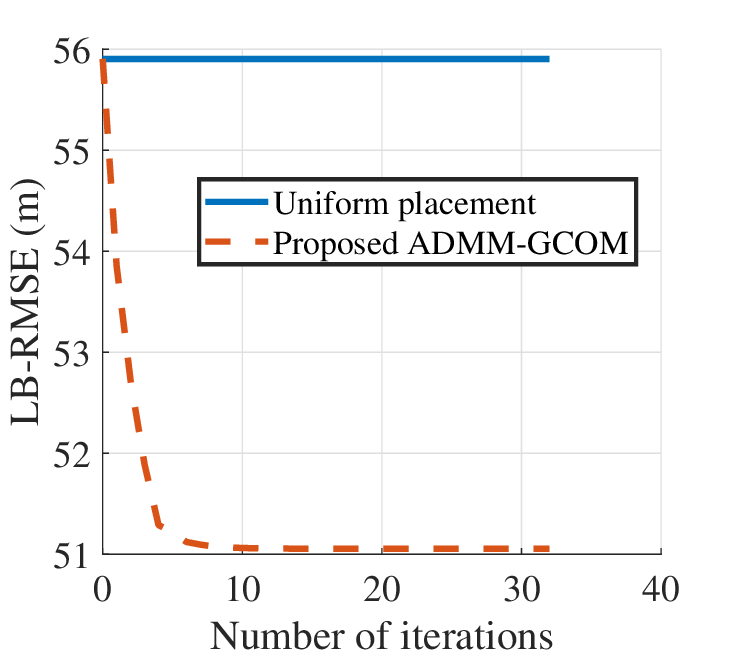}
     }
     \\
      \subfloat[ \label{fig_s3}$\beta_{max}=280^{o}$.$\bar{k}_{MM}=4$.]{%
       \includegraphics[width=0.235\textwidth]{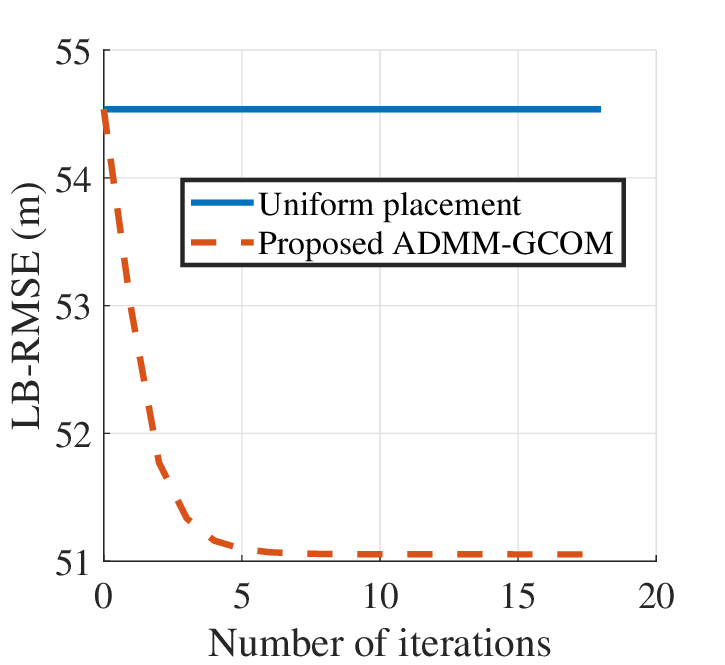}
     }
     \hfill
     \subfloat[\label{fig_s4}$\beta_{max}=360^{o}$.$\bar{k}_{MM}=4$.]{%
       \includegraphics[width=0.235\textwidth]{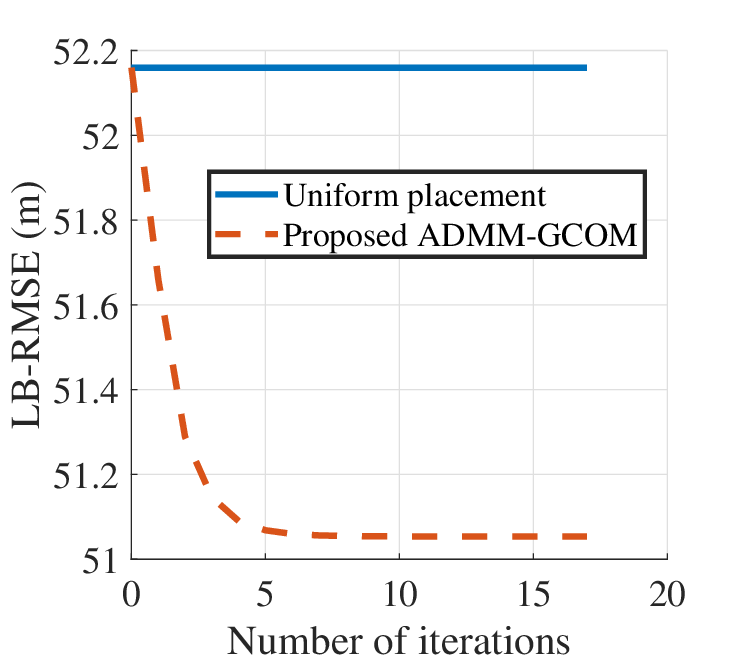}
     }
     \caption{Convergence plots with different constrained angles in case A.}
     \label{Cp_A}
\end{figure}

\begin{figure}
   \subfloat[ \label{fig_s5}$\beta_{max}=120^{o}$.$\bar{k}_{MM}=2$.]{%
       \includegraphics[width=0.235\textwidth]{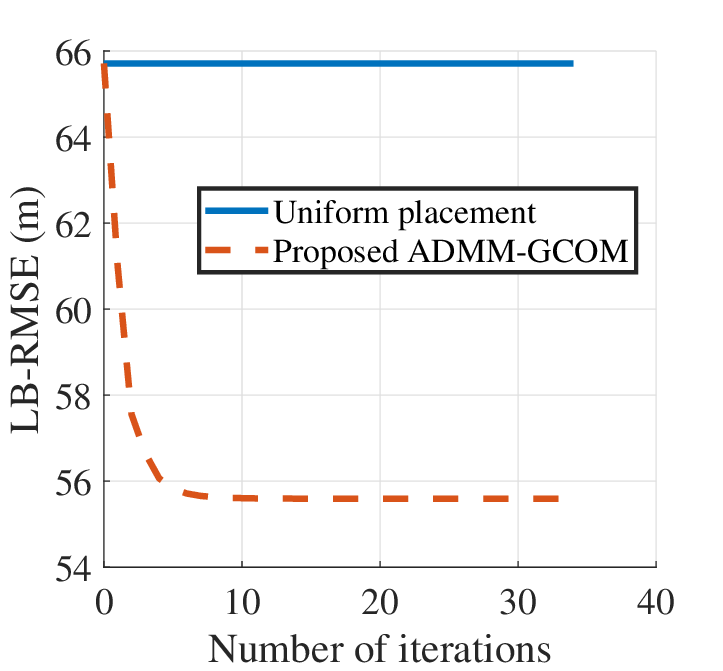}
     }
     \hfill
     \subfloat[\label{fig_s6}$\beta_{max}=200^{o}$.$\bar{k}_{MM}=2$.]{%
       \includegraphics[width=0.235\textwidth]{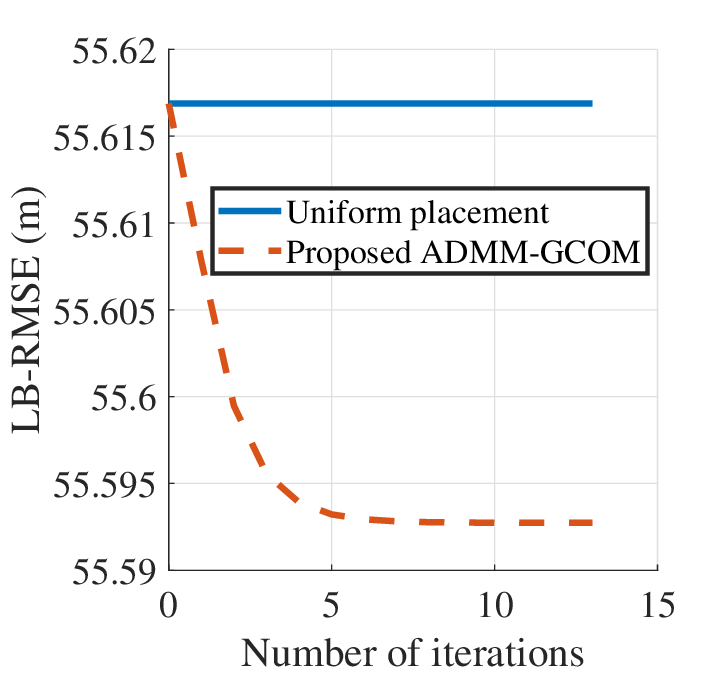}
     }
     \\
      \subfloat[ \label{fig_s7}$\beta_{max}=280^{o}$.$\bar{k}_{MM}=2$.]{%
       \includegraphics[width=0.235\textwidth]{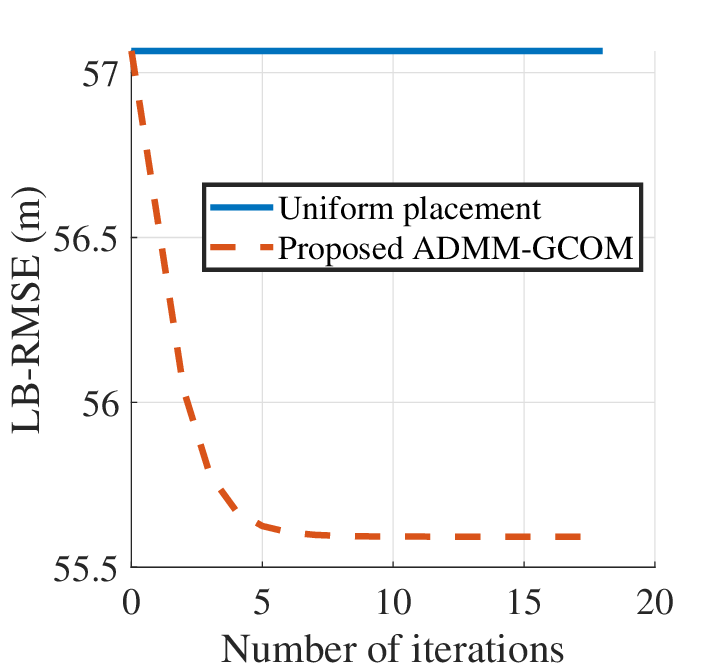}
     }
     \hfill
     \subfloat[\label{fig_s8}$\beta_{max}=360^{o}$.$\bar{k}_{MM}=2$.]{%
       \includegraphics[width=0.235\textwidth]{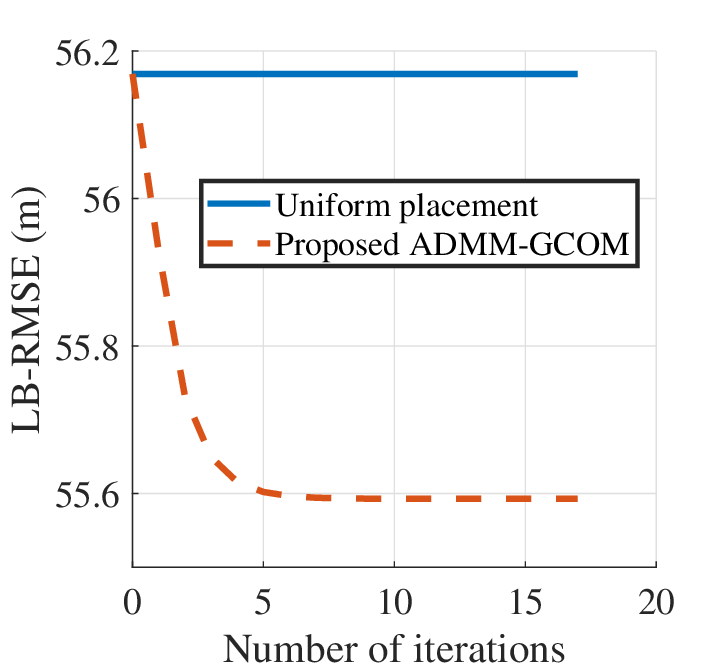}
     }
     \caption{Convergence plots with different constrained angles in case B.}
     \label{Cp_B}
\end{figure}

This section presents simulation results to evaluate the performance of the proposed ADMM-CGCOM, compared to  the commonly used uniform placement strategy \cite{zhao2013optimal,heydari2020optimal}.

Without loss of generality, the default settings are shown as follows:  $N=8$, $\gamma=2$, $r_{i}=1000~\mathrm{m}$, $h_{i}=100~\mathrm{m},~\forall i$. To make the simulation comprehensive, two typical cases named case A and case B are considered. In case A, the measurement conditions of UAVs are different. Specifically, half of the measuring noise variances are equal to $8~\mathrm{dBm}$ while the others are equal to $2~\mathrm{dBm}$,  i.e., $\sigma_{i}^2=8,~i=1,\cdots,\frac{N}{2}$ and $\sigma_{i}^2=2,~i=\frac{N}{2}+1,\cdots,N$. In case B, the measurement conditions of UAVs are identical. Specifically, all measuring noise variance are equal to $4~\mathrm{dBm}$ ,  i.e.,$\sigma_{i}^2=4,~i=1,\cdots,N$. In addition, each UAV samples $10$ times quickly at the same measuring position and then averages the sampled RSSs to get the final measurement value.

Under the above settings, the proposed ADMM-CGCOM is used to optimize the geometric configuration of measuring UAVs, namely $\{\beta_{i}\}_{i=1}^{N}$. To evaluate the optimized geometric configuration, a widely used lower bound of RMSE (abbreviated as LB-RMSE) is used as the performance metric of geometric configuration\footnote{LB-RMSE can be viewed as the lowest average distance error for any unbiased estimation.  It is more intuitive and straightforward to evaluate geometric configurations than the determinant of FIM.}. Mathematically, LB-RMSE $=\sqrt{\mathbf{F}^{-1}(2,2)+\mathbf{F}^{-1}(3,3)}$. Moreover, the convergence threshold of ADMM is $10^{-4}$ and the convergence threshold of MM is $10^{-3}$ in the simulation.

\begin{figure}
   \subfloat[ \label{fig_s9}$\beta_{max}=120^{o}$.]{%
       \includegraphics[width=0.235\textwidth]{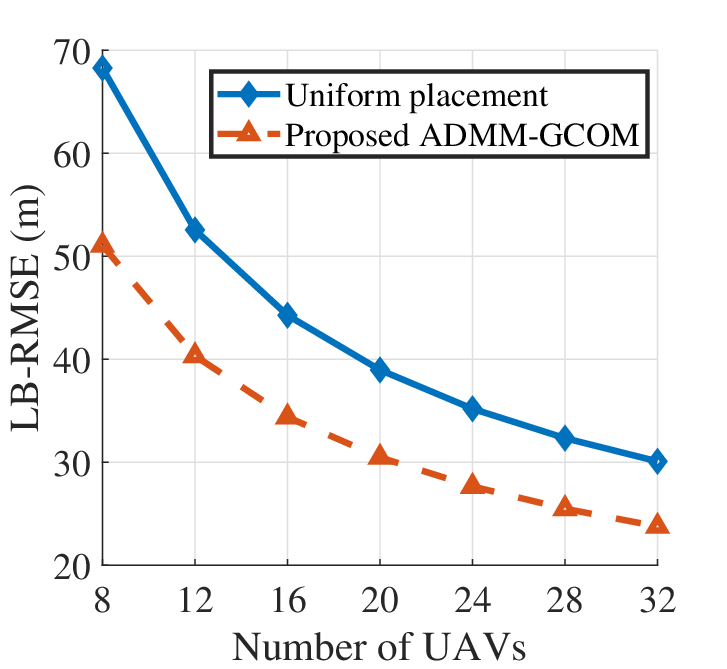}
     }
     \hfill
     \subfloat[\label{fig_s10}$\beta_{max}=200^{o}$.]{%
       \includegraphics[width=0.235\textwidth]{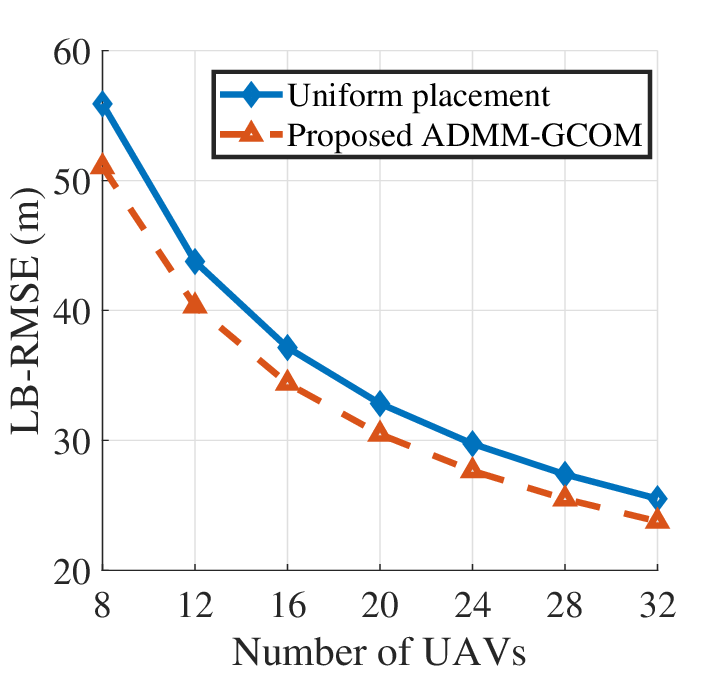}
     }
     \\
      \subfloat[ \label{fig_s11}$\beta_{max}=280^{o}$.]{%
       \includegraphics[width=0.235\textwidth]{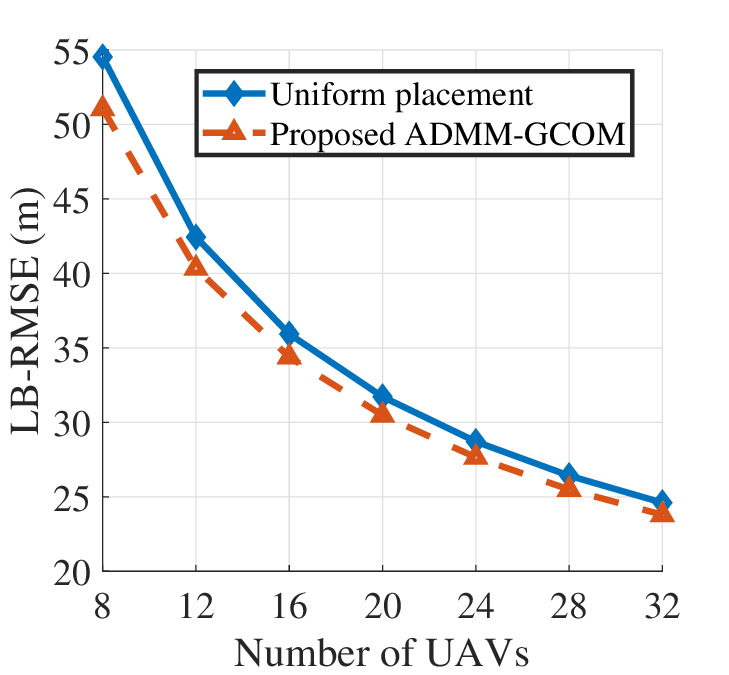}
     }
     \hfill
     \subfloat[\label{fig_s12}$\beta_{max}=360^{o}$.]{%
       \includegraphics[width=0.235\textwidth]{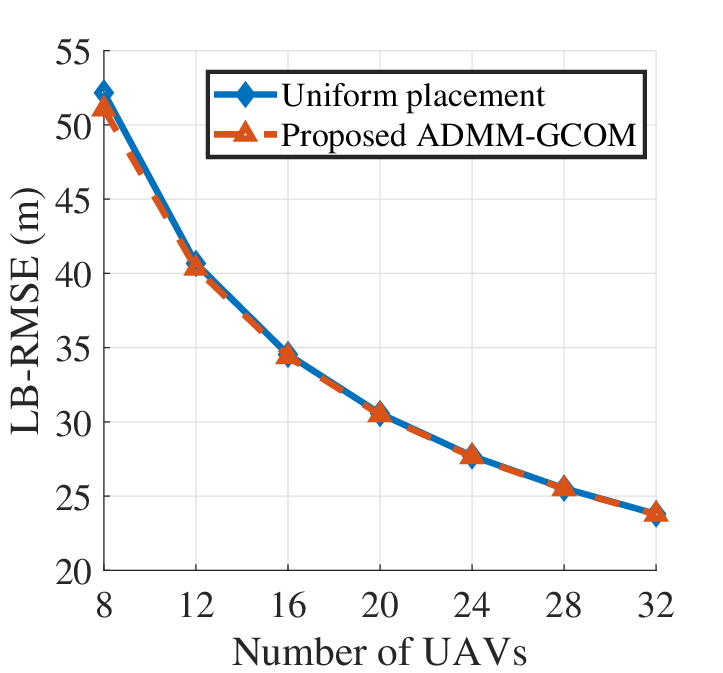}
     }
     \caption{LB-RMSE versus the number of UAVs with different constrained angles in case A.}
     \label{RUAV_A}
\end{figure}
\begin{figure}
   \subfloat[ \label{fig_s13}$\beta_{max}=120^{o}$.]{%
       \includegraphics[width=0.235\textwidth]{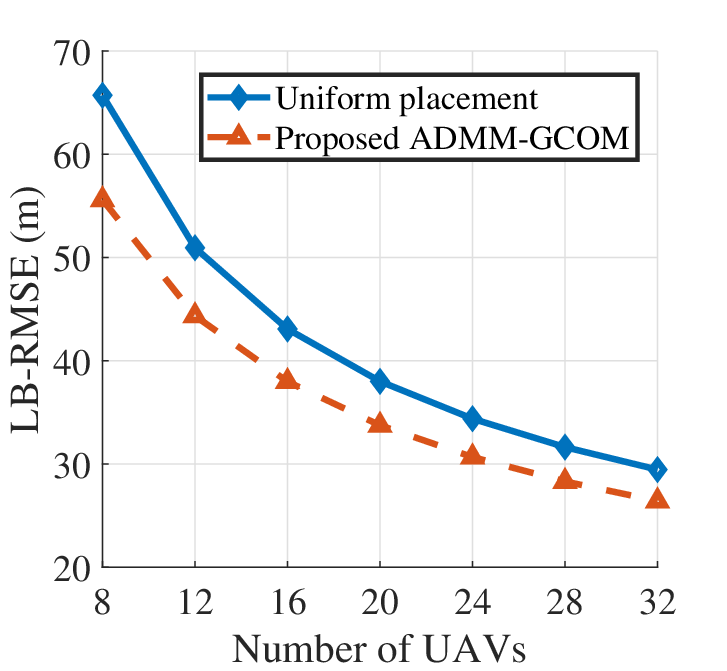}
     }
     \hfill
     \subfloat[\label{fig_s14}$\beta_{max}=200^{o}$.]{%
       \includegraphics[width=0.235\textwidth]{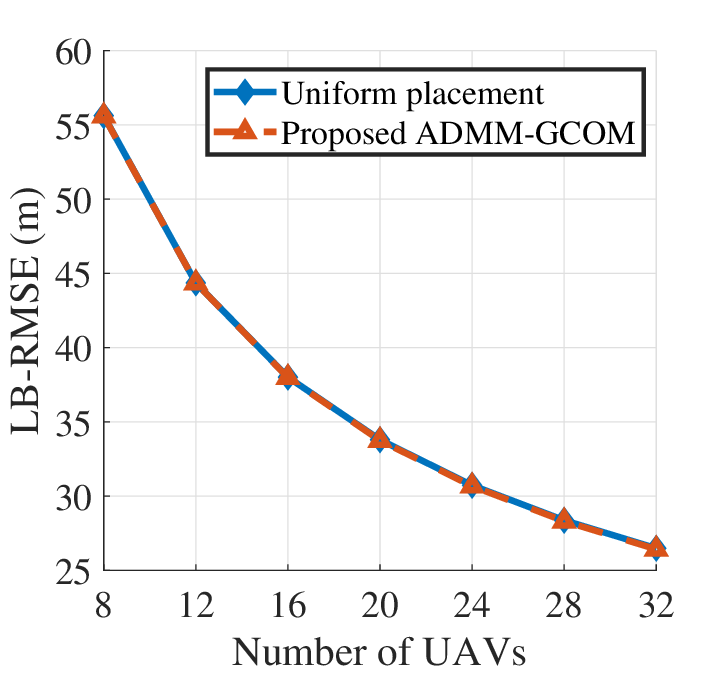}
     }
     \\
      \subfloat[ \label{fig_s15}$\beta_{max}=280^{o}$.]{%
       \includegraphics[width=0.235\textwidth]{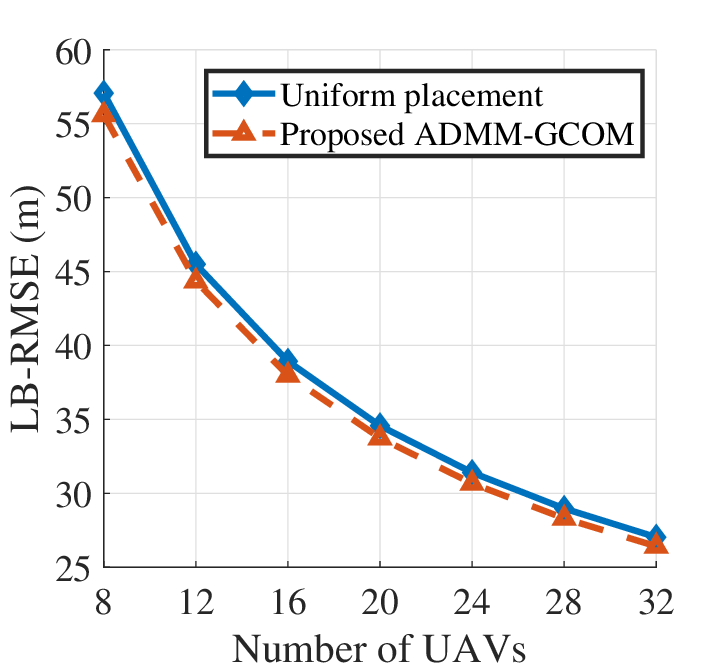}
     }
     \hfill
     \subfloat[\label{fig_s16}$\beta_{max}=360^{o}$.]{%
       \includegraphics[width=0.235\textwidth]{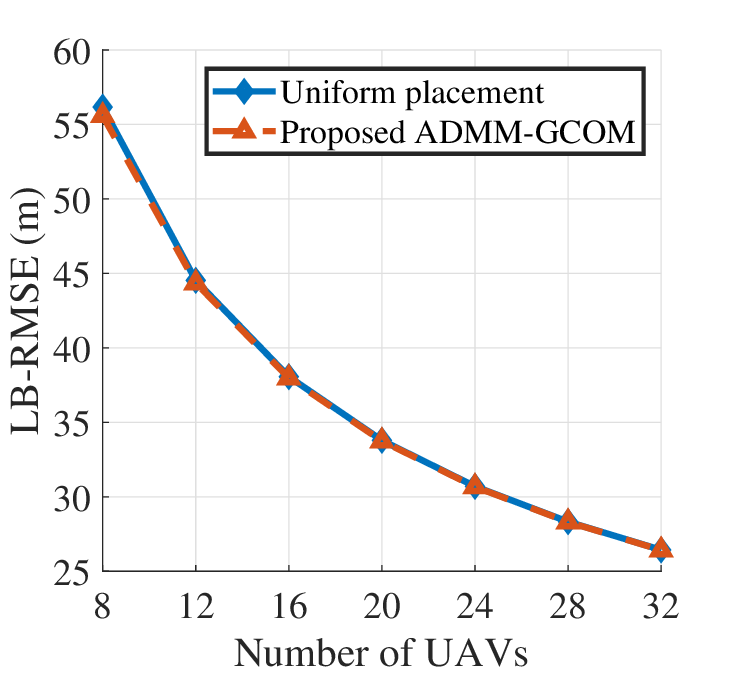}
     }
     \caption{LB-RMSE versus the number of UAVs with different constrained angles in case B.}
     \label{RUAV_B}
\end{figure}

Fig.~\ref{Cp_A} shows the convergence curve of the proposed ADMM-CGCOM under different constrained angles in case A. As seen, starting from the uniform placement strategy, the proposed algorithm can further reduce LB-RMSE, thereby improving the performance of localization system. In addition, the improvement becomes obvious just after a few iterations. For example, after $10$ iterations, the proposed algorithm reduces LB-RMSE by nearly $6\%$ when $\beta_{max}=280^{o}$ and reduces LB-RMSE by nearly $25\%$ when $\beta_{max}=120^{o}$. Fig.~\ref{Cp_A} also shows that under different constrained angles, the proposed algorithm converges within $100$ iterations while the average iteration time of the internal MM is no more than $4$.

Fig.~\ref{Cp_B} shows the convergence curve of the proposed ADMM-CGCOM under different constrained angles in case B. Similar to case A as shown in Fig.~\ref{Cp_A}, from Fig.~\ref{Cp_B}, it is concluded that the proposed algorithm improves the performance of the localization system compared to the uniform placement strategy obviously just after a few iterations. Moreover, it is found that the improvement ratio of the proposed algorithm in case A (inconsistent measure conditions) is more significant than case B (identical measure conditions).

Fig.~\ref{RUAV_A} presents the performance of the proposed ADMM-CGCOM versus the number of UAVs under different constrained angles in case A while  Fig.~\ref{RUAV_B} presents that in case B. From Fig.~\ref{RUAV_A} and Fig.~\ref{RUAV_B},  it is observed that LB-RMSE decreases largely with increasing number of UAVs under different constrained angles. Moreover, it is seen that the proposed algorithm is over the uniform placement strategy with different scales of drone swarm, and the improvement ratio of the proposed algorithm with respect to the uniform placement strategy basically does not change versus the number of UAVs under different constrained angles.

Theoretically, when using the proposed ADMM-CGCOM, the real position of source should be known. It is unrealistic. But with just small adjustment, the proposed algorithm can be used to refine a rough prior estimator in the practical scenario. Specifically, the prior estimator is used to replace the real position of source when using ADMM-CGCOM. In this way, the measuring positions of UAVs are obtained by ADMM-CGCOM, then UAVs can make effective measurement and estimation. Simulations about this practical scenario are shown in  Fig.~\ref{Angle_A} and Fig.~\ref{Angle_B}.

Fig.~\ref{Angle_A} presents the performance of the proposed ADMM-CGCOM versus the constrained angle considering both theoretical and practical scenarios in case A while Fig.~\ref{Angle_B} presents that in case B. In the practical scenario, the prior estimator is unbiased and the estimation error follows gaussian distribution with standard deviation equal to $\sqrt{12500}~\mathrm{m}$. As observed from Fig.~\ref{Angle_A} and Fig.~\ref{Angle_B}, the proposed algorithm is over the uniform placement strategy. Moreover, when the constrained angle exceeds a threshold ($\beta_{max}>=97.5^{o}$ in case A and $\beta_{max}>=105^{o}$ in case B), the performance of the proposed algorithm is basically unchanged as the constrained angle increases. Therefore, the proposed algorithm does not require a large constrained angle to achieve a considerable performance. Finally, as seen in Fig.~\ref{Angle_A} and Fig.~\ref{Angle_B}, in the practical scenario, although the standard deviation of prior estimation is large, the performance of the proposed algorithm almost does not deviate from the theoretical value. Therefore the prior estimation can be refined well by the proposed algorithm.

\begin{figure}[!ht]
  \centering
  \includegraphics[width=0.5\textwidth]{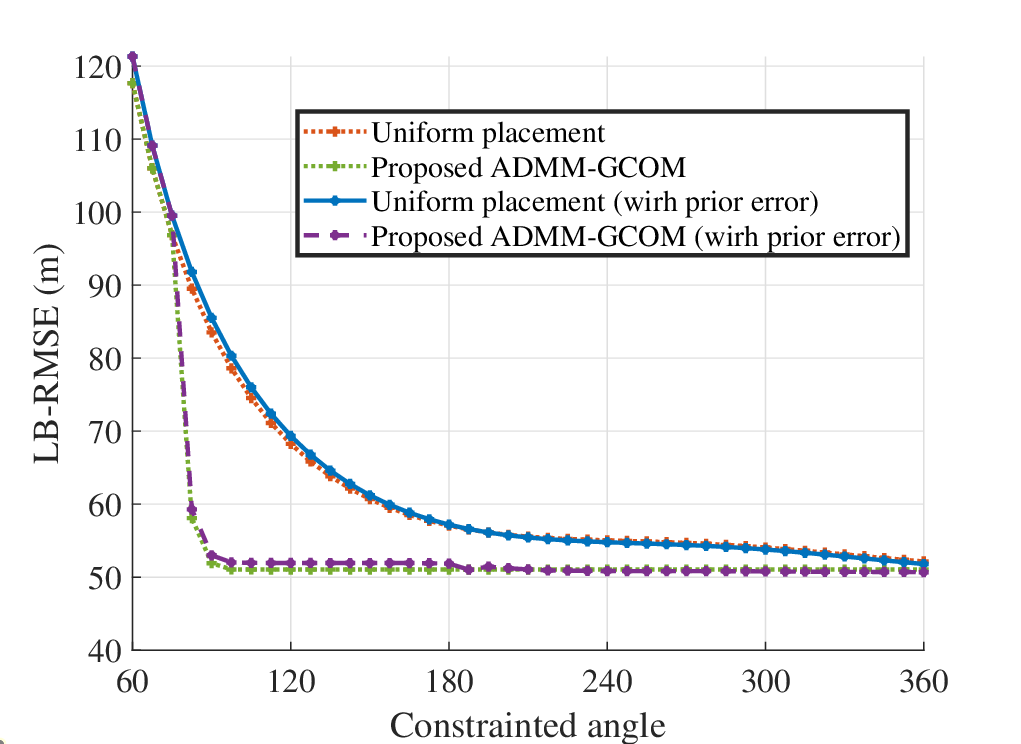}\\
  \caption{~LB-RMSE versus the constrained angle in case A.}\label{Angle_A}
\end{figure}

\begin{figure}[!ht]
  \centering
  \includegraphics[width=0.5\textwidth]{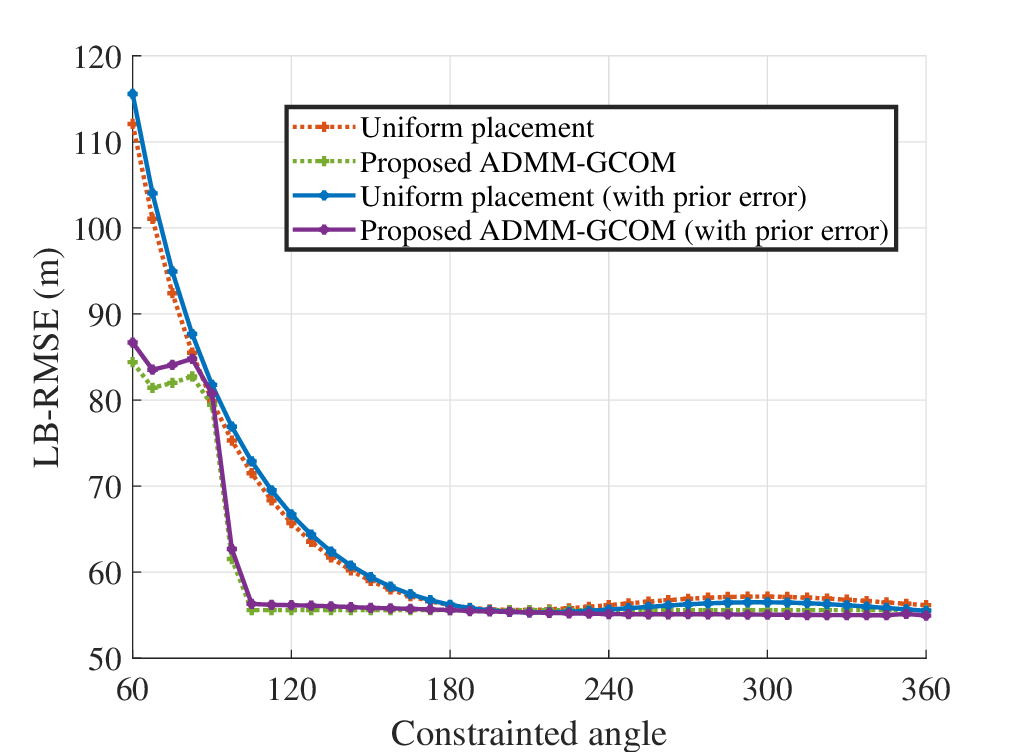}\\
  \caption{~LB-RMSE versus the constrained angle in case B.}\label{Angle_B}
\end{figure}

\section{Conclusion}
In this paper we have investigated the geometric configuration of RSSD-based passive source localization system by drone swarm with spread angle constraint. We have formulated the constrained geometric configuration optimization problem under the D-optimization criterion, and transformed it from a multi-scalar form to a single-matrix matrix form, then proposed a numerical optimization method named ADMM-CGCOM. The proposed ADMM-CGCOM has been compared with the uniform placement strategy in simulations, where we have shown that ADMM-CGCOM leads to  a lower LB-RMSE of the final estimation. Moreover, the function of ADMM-CGCOM on refining the initial estimation in the practical scenario has been demonstrated. Using the proposed ADMM-CGCOM, the effective measuring positions of drone swarm with constrained spread angle can be obtained.

In practical applications, the locating area of a draw swarm may contain multiple interested signal sources, making it challenging to directly apply the proposed algorithm. Therefore, it is worth further investigating how to distinguish the RSS of multiple targets and how to optimize the geometric configuration of drone swarm when comprehensively considering the positioning errors of multiple targets. Furthermore, in practical scenarios, the prior estimation is directly employed to substitute the true position of the target during the optimization of the geometric configuration. To enhance the measuring effectiveness and mitigate the risk of violating actual measuring region constrains, further research is essential to develop more sophisticated models and to advance the proposed algorithm by leveraging the statistic characteristics of prior estimation.

\appendices
\section{Proof of Theorem \ref{thm1}}\label{appendixa}
Let $\overline{\lambda}$ denote the mean of eigenvalues of $\mathbf{T}$. From Eq. (\ref{deT}), it shows that $\mathbf{T}\in \mathbb{R}^{2\times2}$ and $\mathbf{T}$ is a symmetric matrix. Therefore,
\begin{align}\label{tezhengT}
\left|\mathbf{T}\right|
&=-\left\|\mathbf{T}\right\|^{2}+2\overline{\lambda}^{2} \nonumber\\
&=-\left\|\mathbf{T}\right\|^{2}+\frac{1}{2}\left(\sum_{i=1}^{N}w_{i}\frac{r_{i}^2}{d_{i}^4}-\left\|\sum_{i=1}^{N}w_{i}\frac{r_{i}}{d_{i}^2} \mathbf{g}_{i}\right\|^{2}\right)^2.
\end{align}

Let $\mathbf{U}\in \mathbb{R}^{2\times2}$ represent an orthogonal matrix. For an arbitrary geometric configuration denoted as $\{\mathbf{g}_{i}\}_{i=1}^{N}$, we make an orthogonal transformation by $\mathbf{U}$ to obtain a new geometric configuration,
denoted as $\{\mathbf{g^{'}}_{i}\}_{i=1}^{N}$, i.e., $\mathbf{g^{'}}_{i}=\mathbf{U}\mathbf{g}_{i},~\forall i$. According to Eq. (\ref{deT}), it shows that $\mathbf{T}\left(\{\mathbf{g^{'}}_{i}\}_{i=1}^{N}\right)=\mathbf{U}\mathbf{T}
\left(\{\mathbf{g}_{i}\}_{i=1}^{N}\right)\mathbf{U}^{T}$. Then, we have
\begin{align}\label{eqt1}
&\|\mathbf{T}\left(\{\mathbf{g^{'}}_{i}\}_{i=1}^{N}\right)\|^{2}_{2} \nonumber\\
&=\mathrm{Tr}\left(\mathbf{U}\mathbf{T}
\left(\{\mathbf{g}_{i}\}_{i=1}^{N}\right)\mathbf{U}^{T}\mathbf{U}\mathbf{T}
\left(\{\mathbf{g}_{i}\}_{i=1}^{N}\right)^{T}\mathbf{U}^{T}\right) \nonumber\\
&=\|\mathbf{T}\left(\{\mathbf{g}_{i}\}_{i=1}^{N}\right)\|^{2}.
\end{align}
Similarly,
\begin{align}\label{eqt2}
\left\|\sum_{i=1}^{N}w_{i}\frac{r_{i}}{d_{i}^2}\mathbf{g}_{i}^{'}\right\|^{2}=\left\|\sum_{i=1}^{N}w_{i}\frac{r_{i}}{d_{i}^2} \mathbf{g}_{i}\right\|^{2}.
\end{align}

According to Eq. (\ref{tezhengT}),  Eq. (\ref{eqt1}) and  Eq. (\ref{eqt2}), we have $|\mathbf{T}^{'}|=|\mathbf{T}|$. Therefore, the effectiveness of geometric configuration (the determinant of FIM) is unchanged after the orthogonal transformation. Moreover, referring  \cite{zhao2013optimal},  orthogonal transformation means horizontal overall rotation or horizontal overall with respect to the source in the focused localization system.

\section{The Global Optimal Property of $\mathbf{G}^{k}_{T}$}\label{proofMM}
According to the definition of $\mathbf{M}$, $L_{\rho}^{k}$ in Eq. (\ref{lossG0}) can be expressed as
\begin{align}
L_{\rho}^{k}
&=\frac{\rho}{2}\mathrm{Tr}\left(\mathbf{G}^{T}\mathbf{M}\mathbf{G}  \right)+\mathrm{Tr}\left(\mathbf{Q}^{kT}\mathbf{G}\right)+\alpha^{k} \nonumber\\
&=\frac{\rho}{2}\sum_{i=1}^{N}\mathbf{g}_{i}^{T}\mathbf{M}\mathbf{g}_{i}+\sum_{i=1}^{N}\mathbf{q}_{i}^{kT}\mathbf{g}_{i},
\end{align}
where $\mathbf{Q}^{kT}=\mathbf{C}^{kT}\mathbf{B}^{\frac{1}{2}}\mathbf{D}$ and $\mathbf{q}_{i}^{kT}$ is the $i$-th row of $\mathbf{Q}^{kT}$. The second-order derivative of $L_{\rho}^{k}$ with respect to $\mathbf{g}_{i}$ is given by
\begin{align}
\nabla_{\mathbf{g}_{i}}^{2}L_{\rho}(\mathbf{X}^{k+1},\mathbf{G},\mathbf{V}^{k})=\rho\mathbf{M}\succeq \mathbf{0},~\forall i.
\end{align}
Therefore, the objective function of subproblem (\ref{step2})) is a convex function with respect to $\mathbf{g}_{i}$, $\forall i$.  But the constraint $\mathcal{D}$ in subproblem (\ref{step2}) is not convex. Therefore, the convergence of MM may be local optimal. However, the global optimal property of it is demonstrated by the following analysis.

First, we slack $\mathcal{D}$ to generate a new set denoted as $\overline{\mathcal{D}}$  as follows.
\begin{align}
\overline{\mathcal{D}}=\{\mathbf{G}|\mathbf{g}_{i}^{T}\mathbf{g}_{i}\leq1,~\mathbf{g}_{i}\geq \mathbf{g}_{0},~\forall i\}.
\end{align}
Assume $\mathbf{G}_{e}\in \overline{\mathcal{D}}$,$\mathbf{G}_{f}\in \overline{\mathcal{D}}$. Let $\mathbf{g}_{e,i}$ denote the transpose of the $i$-th row of $\mathbf{G}_{e}$ and $\mathbf{g}_{f,i}$ denote the transpose of the $i$-th row of $\mathbf{G}_{f}$. For any $(\alpha, \beta)$ satisfying $0\leq\alpha\leq1 \cup 0\leq\beta\leq1 \cup \alpha+\beta=1$, we have
 \begin{align}
\alpha\mathbf{g}_{e,i}+\beta \mathbf{g}_{f,i}\geq  \begin{bmatrix} (\alpha+\beta)  \mathbf{g}_{0}(1) & (\alpha+\beta) \mathbf{g}_{0}(2)\end{bmatrix}     \geq \mathbf{g}_{0},
\end{align}
 \begin{align}
&(\alpha\mathbf{g}_{e,i}+\beta \mathbf{g}_{f,i})^{T}(\alpha\mathbf{g}_{e,i}+\beta \mathbf{g}_{f.i})\nonumber\\
&=\alpha^2\|\mathbf{g}_{e,i}\|^{2}+\beta^2\|\mathbf{g}_{f,i}\|^{2}+2\alpha\beta\mathbf{g}_{e,i}^{T}\mathbf{g}_{f,i}  \leq (\alpha+\beta)^{2}=1.
\end{align}
Therefore, $\overline{\mathcal{D}}$ is a convex set.

Then, taking $\overline{\mathcal{D}}$ as the constraint, a new subproblem based on subproblem (\ref{step2}) is considered as follows.
\begin{align}\label{pro22}
\mathbf{G}^{k+1}=\mathop{\mathrm{argmin}}_{\mathbf{G}\in \overline{\mathcal{D}}}  L_{\rho}(\mathbf{X}^{k+1},\mathbf{G},\mathbf{V}^{k}).
\end{align}
From above, it shows that subproblem (\ref{pro22}) is convex.

Finally, compare the two subproblems. When using the MM algorithm in Section \ref{7semm}, it is easy to verify that the obtained solution of subproblem (\ref{pro22}) ($\mathbf{G}^{k}_{T}$) is the same as that in subproblem (\ref{step2}). Note that $\mathbf{G}^{k}_{T}$ must be a local optimal solution of subproblem (\ref{pro22}). And it is  well known that the local optimal solution of a convex problem is also the global optimal solution. Therefore, $\mathbf{G}^{k}_{T}$ is global optimal on subproblem (\ref{pro22}). Since $\mathcal{D}\in\overline{\mathcal{D}}$, subproblem (\ref{pro22}) is a slack version of subproblem (\ref{step2}), it can be concluded   that $\mathbf{G}^{k}_{T}$ is also the  global optimal solution of subproblem (\ref{step2}).

\section{Analysis For  More Practical Scenarios}\label{proofre3}
Consider more practical scenarios where the UAV-target horizontal distance and the UAV-target vertical distance are constrained in addition to horizontal angles.  Let $r_{i}$ represent the horizontal distance  and $h_{i}$ represent the vertical horizontal distance for the $i$-th UAV, respectively. Mathematically, $r_{i}= \sqrt{(x-x_{i})^2+(y-y_{i})^2}$, $h_{i}=z_{i}$, $r_{i}\in \mathcal{D}_{r,i}$, $h_{i}\in \mathcal{D}_{h,i}$, where $\mathcal{D}_{r,i}$ and $\mathcal{D}_{h,i}$ are  constraint sets defined by  the specific scenario.

Based on Section \ref{SeIII}, given an arbitrary UAV-target horizontal angular geometric configuration, the problem of distance optimization is given by
\begin{align}
\mathrm{(P^{'})}:&\max_{\mathbf{D}}~~~~\left|\mathbf{G}^{T} \mathbf{D}^{T} \mathbf{B} \mathbf{D} \mathbf{G}\right|  \nonumber\\
&\text{s. t.}~~~~~~~r_{i}\in \mathcal{D}_{r,i},\\ \nonumber
&~~~~~~~~~~~h_{i}\in \mathcal{D}_{h,i},~~~i=1,\cdots,N.
\end{align}
Here, $\mathbf{D}\triangleq\mathrm{diag}\left(\begin{bmatrix} \frac{r_{1}}{d_{1}^2} & \frac{r_{2}}{d_{2}^2}& \cdots & \frac{r_{N}}{d_{N}^2} \end{bmatrix}\right)$, where $d_{i}=\sqrt{r_{i}^2+h_{i}^2},~\forall i$.

Since $\mathbf{B}$ is a semi-positive definite matrix, shown in Eq. (\ref{psidef}), it can be decomposed by SVD as
\begin{align}
\mathbf{B} = \mathbf{U}_{B} \bm{\Lambda}_{B} \mathbf{U}_{B}^T,
\end{align}
where $\mathbf{U}_{B}$ is an orthogonal matrix and $\bm{\Lambda_{B}}$ is a diagonal matrix  with $\mathbf{\Lambda_{B}}=\mathrm{diag}\left(\begin{bmatrix} \psi_{1} & \psi_{2} &\cdots & \psi_{N} \end{bmatrix}\right)$, $\psi_{i}\geq0,~\forall i$.

Let $\mathbf{S}=\mathbf{D}\mathbf{B}\mathbf{D}$. Since $\mathbf{D}$ is a diagonal matrix with only positive diagonal elements,   $\mathbf{S}$ is a semi-positive definite matrix, which can be decomposed by SVD as
\begin{align}
\mathbf{S} = \mathbf{U}_{S} \bm{\Lambda}_{S} \mathbf{U}_{S}^T,
\end{align}
where $\mathbf{U}_{S}$ is the orthogonal matrix and $\Lambda_{S}$ is a diagonal matrix. It is straightforward that
\begin{align}
&\mathbf{U}_{S}=\mathbf{U}_{B}, \\ \nonumber
&\bm{\Lambda}_{S}=\mathbf{D}\bm{\Lambda}_{B}\mathbf{D}.
\end{align}

Based on the above, the objective function of problem $\mathrm{P^{'}}$ can be expressed as
\begin{align}\label{objopddri}
\left|\mathbf{G}^{T}\mathbf{D}^{T}\mathbf{B}\mathbf{D}\mathbf{G}\right|
&=\left|(\mathbf{G}^{T}\mathbf{U}_{B})(\mathbf{D}\bm{\Lambda}_{B}\mathbf{D})(\mathbf{U}_{B}^{T}\mathbf{G})\right|\\\nonumber
&=\left|\mathbf{G}^{T}\mathbf{U}_{B}\right|^{2}\left|\mathbf{D}\bm{\Lambda}_{B}\mathbf{D}\right| \\\nonumber
&=\left|\mathbf{G}^{T}\mathbf{U}_{B}\right|^{2}\left(\prod_{i=1}^{N}\psi_{i}\frac{r_{i}^{2}}{d_{i}^4}\right).
\end{align}
Since $\psi_{i}\geq0,~\forall i$, the solution to problem $\mathrm{P^{'}}$ is given by
\begin{align}
\{r_{i}^{*},h_{i}^{*}\}=\mathop{\arg\max}\limits_{r_{i}\in \mathcal{D}_{r,i},h_{i}\in \mathcal{D}_{h,i}} \frac{r_{i}}{h_{i}^2+r_{i}^2},~\forall i.
\end{align}

Eq. (\ref{objopddri}) demonstrates that the optimization of horizontal angle and the optimization of distance are independent, validating that  Algorithm \ref{alg_1} is applicable and effective in these practical scenarios.

\ifCLASSOPTIONcaptionsoff
  \newpage
\fi

\bibliographystyle{IEEEtran}
\bibliography{IEEEfull,cite}

\end{document}